%% file: main.tex
\documentclass[a4paper,USenglish,cleveref,autoref,thm-restate]{lipics-v2021}

\usepackage{amsthm,amssymb,amsmath}  
\usepackage{mathtools}
\usepackage{graphicx}
\usepackage{xspace}
\usepackage{todonotes}
\usepackage{tikz}
\usepackage{forest}
\usepackage{comment}
\usepackage{pgfplots}
\usepackage{pifont}
\pgfplotsset{compat=1.18}

\usetikzlibrary{shapes.geometric, arrows.meta, positioning, patterns}

\let\oldtextsc\textsc
\renewcommand{\textsc}[1]{\textup{\oldtextsc{#1}}}

\newtheorem{fact}[theorem]{Fact}

\newcommand{\cO}{\mathcal{O}}
\def\dd{\mathinner{.\,.}}
\newcommand{\CT}{\textsf{\textup{CT}}}

\newcommand{\absolute}[1]{\lvert#1\rvert}

\newcommand{\floor}[1]{\lfloor#1\rfloor}

\newcommand{\ceil}[1]{\lceil#1\rceil}

\def\hdist{\textnormal{\textsf{CHd}}}
\def\dphd{\textnormal{\textsf{optH}}}
\def\dphdsub{\textnormal{\textsf{optHSub}}}
\def\dphdnosub{\textnormal{\textsf{optHNoSub}}}
\def\problemname{\textsc{Approximate CT-Matching with Substitutions}\xspace}

\def\psv{\textnormal{\textsf{psv}}}
\def\nsv{\textnormal{\textsf{nsv}}}

\def\PD{\textnormal{\textsf{PD}}}
\def\ND{\textnormal{\textsf{ND}}}
\def\per{\textnormal{per}}

\newlength{\probboxminipagewidth}
\newlength{\fboxsepbak}
\newcommand{\defproblem}[4][%
]{%
\setlength{\fboxsepbak}{\fboxsep}
\setlength{\fboxsep}{.5em}
\setlength{\probboxminipagewidth}{\textwidth - 2\fboxsep - 2\fboxrule}%
\par\vspace{\topsep}\noindent\fbox{%
  \begin{minipage}{\probboxminipagewidth}
    {#1{{\boldmath\large#2\unboldmath}}} \par\smallskip
    {\bf{Input:}} #3 \par\smallskip
    {\bf{Output:}} #4
  \end{minipage}%
  }%
\par\vspace{\topsep}%
\setlength{\fboxsep}{\fboxsepbak}%
}

\let\dagstuhlifnum\ifnum

\title{Approximate Cartesian Tree Matching\texorpdfstring{\\}{\ }with Substitutions}
\titlerunning{Approximate Cartesian Tree Matching with Substitutions}



\category{} 

\acknowledgements{We thank the anonymous reviewers of STACS 2026 for their helpful feedback.}

\nolinenumbers 

\EventEditors{Meena Mahajan, Florin Manea, Annabelle McIver, and Nguy\~{\^{e}}n Kim Th\'{\u{a}}ng}
\EventNoEds{4}
\EventLongTitle{43rd International Symposium on Theoretical Aspects of Computer Science (STACS 2026)}
\EventShortTitle{STACS 2026}
\EventAcronym{STACS}
\EventYear{2026}
\EventDate{March 10--13, 2026}
\EventLocation{Grenoble, France}
\EventLogo{}
\SeriesVolume{364}
\ArticleNo{46}

\hideLIPIcs

\author{Panagiotis Charalampopoulos}{King's College London, UK}{p.charalampopoulos@kcl.ac.uk}{https://orcid.org/0000-0002-6024-1557}{}
\author{Jonas Ellert}{DIENS, \'{E}cole normale sup\'{e}rieure de Paris, PSL Research University, France}{ellert.jonas@gmail.com}{https://orcid.org/0000-0003-3305-6185}{Partially funded by grant ANR20-CE48-0001 from the French National Research Agency.}
\author{Manal Mohamed}{King's College London, UK}{manal.1.mohamed@kcl.ac.uk}{https://orcid.org/0000-0002-1435-5051}{}

\authorrunning{P. Charalampopoulos, J. Ellert, M. Mohamed}
\Copyright{Panagiotis Charalampopoulos, Jonas Ellert, and Manal Mohamed} 

\begin{document}

\maketitle

\begin{abstract}
The Cartesian tree of a sequence captures the relative order of the sequence's elements.
In recent years, Cartesian tree matching has attracted considerable attention, particularly due to its applications in time series analysis.
Consider a text $T$ of length $n$ and a pattern $P$ of length~$m$. 
In the exact Cartesian tree matching problem, the task is to find all length-$m$ fragments of $T$ whose Cartesian tree coincides with the Cartesian tree $\CT(P)$ of the pattern.
Although the exact version of the problem can be solved in linear time [Park et al., TCS 2020], it remains rather restrictive; for example, it is not robust to outliers in the pattern.

To overcome this limitation, we consider the approximate setting, where the goal is to identify all fragments of $T$ that are \emph{close} to some string whose Cartesian tree matches $\CT(P)$.
In this work, we quantify \emph{closeness} via the widely used Hamming distance metric.
For a given integer parameter $k>0$, we present an algorithm that computes all fragments of~$T$ that are at Hamming distance at most $k$ from a string whose Cartesian tree matches $\CT(P)$.
Our algorithm runs in time $\cO(n \sqrt{m} \cdot k^{2.5})$ for $k \leq m^{1/5}$ and in time $\cO(nk^5)$ for $k \geq m^{1/5}$, thereby improving upon the state-of-the-art $\cO(nmk)$-time algorithm of Kim and Han [TCS 2025] in the regime $k = o(m^{1/4})$.

On the way to our solution, we develop a toolbox of independent interest.
First, we introduce a new notion of periodicity in Cartesian trees. Then, we lift multiple well-known combinatorial and algorithmic results for string matching and periodicity in strings to Cartesian tree matching and periodicity in Cartesian trees.
\end{abstract}

\clearpage
\setcounter{page}{1}

\section{Introduction}

Pattern matching is a central topic in theoretical computer science and a ubiquitous task in practical applications.
In several real-world scenarios, such as stock price analysis and music analysis, one may wish to match a pattern with respect to the relative order between its sequential data points rather than their exact values.
For instance, consider the challenge of matching the performance of a stock in the past quarter to its historical performance in order to predict its future behaviour.
In this task, we are not concerned with the exact stock price but rather with its ``ups and downs''.
Several notions of matching have been devised for capturing such scenarios, including Cartesian tree matching~\cite{DBLP:journals/corr/abs-2505-09236,DBLP:conf/cpm/KimC21,DBLP:conf/spire/NishimotoFNI21,DBLP:conf/cpm/OsterkampK25,exactCTmatching,DBLP:journals/tcs/SongGRFLP21} and order-preserving pattern matching~\cite{DBLP:journals/dam/CantoneFK20,DBLP:journals/spe/ChhabraFKT17,DBLP:journals/ipl/ChoNPS15,DBLP:journals/spe/DecaroliGM19,DBLP:conf/cpm/0002HSSTY16,DBLP:journals/tcs/GawrychowskiU16,DBLP:journals/acta/JargalsaikhanHUYS24,DBLP:journals/ipl/KimKNS23,DBLP:conf/spire/NakamuraIBT17,DBLP:journals/iandc/RussoCHBF22}.
Here, we focus on the Cartesian tree matching problem, which was introduced by Park et al.~\cite{exactCTmatching}.

A Cartesian tree $\CT(S)$ of a string $S$ is a binary tree defined recursively as follows: the root of the tree corresponds to the index $i$ of the minimum element in the sequence (under a fixed tie-breaking strategy); the left subtree of the root corresponds to the Cartesian tree of $S[1 \dd i-1]$, and the right subtree of the root corresponds to the Cartesian tree of $S[i+1 \dd \absolute{S}]$.
In what follows, we assume that ties are always resolved in favour of the leftmost elements.
The Cartesian tree of a string $S$ can be built in linear time~\cite{DBLP:journals/jal/BenderFPSS05}.

We say that a string $U$ CT-matches a string $V$ if and only if $\CT(U)=\CT(V)$.
This notion of matching is well-behaved: it is a prime example of substring consistent equivalence relations~\cite{DBLP:journals/tcs/MatsuokaAIBT16} since, if $U$ and $V$ CT-match, then, for any $1\leq i \leq j \leq \absolute{U}$, $U[i\dd j]$ and $V[i \dd j]$ CT-match as well~\cite{exactCTmatching}.

In the Cartesian tree matching (CT-matching) problem, we are given a text $T$ of length~$n$ and a pattern $P$ of length $m$, and the goal is to compute all $m$-length fragments of $T$ that CT-match $P$.
Park et al.~\cite{exactCTmatching} showed that this problem admits a linear-time solution.
An alternative linear-time algorithm, as well as practical solutions, were later presented in~\cite{DBLP:journals/tcs/SongGRFLP21}.
Related problems such as dictionary CT-matching~\cite{exactCTmatching,DBLP:journals/tcs/SongGRFLP21} and indexing for CT-matching~\cite{DBLP:conf/cpm/KimC21,DBLP:conf/spire/NishimotoFNI21,DBLP:conf/cpm/OsterkampK25} have also been studied.

\subparagraph{Approximate Cartesian tree matching.} Similarly to the case of standard string matching, requiring exact CT-matches can be quite restrictive.
For instance, a single outlier value in the pattern (arising, for example, from an extreme but short-lived event or from noise in the data) may cause all relevant occurrences in the text to be missed.
To address this, one can instead seek approximate CT-matches of~$P$ in $T$, that is, fragments of $T$ that are \emph{close} to a string $X$ such that $\CT(X)=\CT(P)$.
Here, \emph{closeness} may be measured in terms of some distance measure such as the Hamming distance or the edit distance.
We consider the low distance regime, where an integer threshold parameter $k > 0$ is provided. Then, we have to report fragments of $T$ at distance at most $k$ from a string $X$ with $\CT(X)=\CT(P)$.

Auvray et al.~\cite{DBLP:journals/corr/abs-2505-09236} recently presented an $\cO(nm)$-time algorithm for the case of one difference (such as a swap or a substitution).
Kim and Han~\cite{DBLP:journals/tcs/KimH25} studied the problem of computing the edit and Hamming distances of two strings under CT-matching.
They presented a dynamic-programming-based solution and, as corollaries, obtained
an $\cO(n^3mk)$-time algorithm for approximate CT-matching under the edit distance
and an $\cO(nmk)$-time algorithm for approximate CT-matching under the Hamming distance.
In this work, we focus on approximate CT-matching under the Hamming distance.

\subparagraph{Our contributions.} For equal-length strings $X$ and $Y$,
let $\hdist(X \leadsto Y)$ denote the minimum number of substitutions needed to transform $X$ into a string $X'$ such that $\CT(X') = \CT(Y)$. We further define $\hdist_k(X \leadsto Y) = \min \{k+1 , \hdist(X \leadsto Y)\}$. We consider the following problem.

\defproblem{\problemname}{A text $T$ of length $n$, a pattern $P$ of length $m$, and an integer threshold $k>0$.}{$\hdist_k(T[i \dd i+m) \leadsto P)$ for all $i\in [1,n-m+1]$.}

The state of the art solution of Kim and Han~\cite{DBLP:journals/tcs/KimH25} computes $\hdist_k(T[i \dd i+m) \leadsto P)$ for each ${i \in [1,n-m+1]}$ separately in $\cO(mk)$ time, resulting in $\cO(nmk)$ overall time.
In the spirit of a line of works in approximate string matching (see, e.g., \cite{DBLP:conf/soda/BringmannWK19,DBLP:conf/focs/Charalampopoulos20}), we show that this is not necessary by transferring the ``few matches or almost periodicity'' paradigm to approximate CT-matching. 
We obtain the following result.

\begin{restatable}{theorem}{mainthm}\label{thm:main}
The \problemname problem can be solved in $\cO(n \sqrt{m} \cdot k^{2.5})$ time for $k \in [1, \floor{m^{1/5}}]$ and in $\cO(nk^5)$ time for any $k \in [\floor{m^{1/5}}, m]$.
\end{restatable}

Hence, we improve upon the state of the art by polynomial factors when $k = \cO(m^{1/4-\epsilon})$ for any constant $\epsilon > 0$.
An illustration of the time complexities of our algorithm and the algorithm of Kim and Han~\cite{DBLP:journals/tcs/KimH25} for the important case when $n=\Theta(m)$ is given in \cref{fig:plot}. Along the way to proving \cref{thm:main}, we develop a toolbox that we expect to find further applications in the study of Cartesian trees.

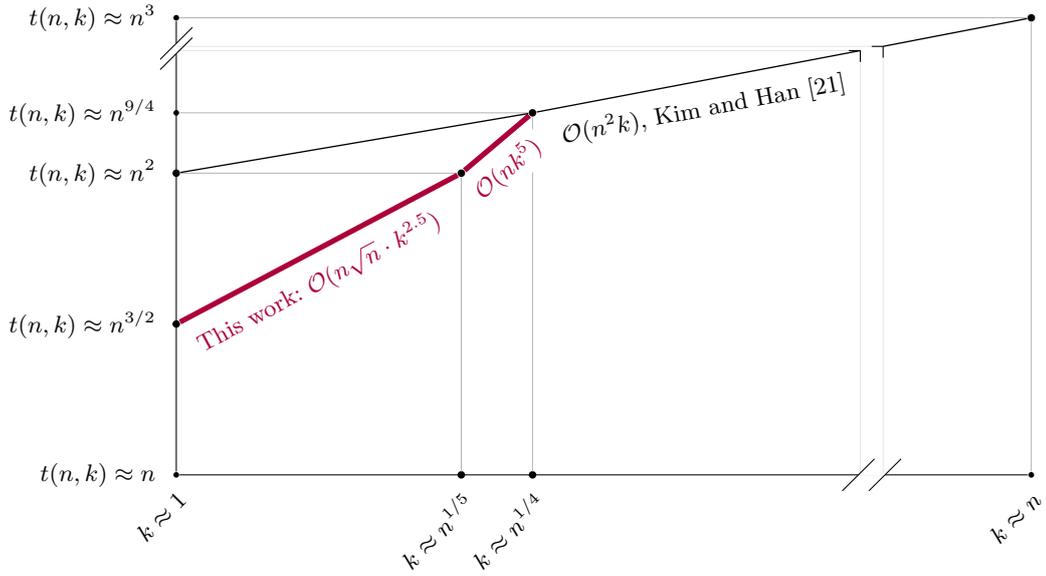
\begin{figure}
\centering\small
{\input{figs/plot.tex}}
\caption{The running time $t(n,k)$ of algorithms for \problemname as a function of $k$ for the special case when $m= \Theta(n)$, shown on a doubly logarithmic~scale.}\label{fig:plot}
\end{figure}

\subparagraph{Technical overview.} 

It is a folklore fact that a pattern $P$ of length $m$ is either periodic or it has at most one exact occurrence in a text $T$ of length $n \leq 3m/2$~\cite{BG95}.
In approximate string matching, recent algorithmic breakthroughs~\cite{DBLP:conf/soda/BringmannWK19,DBLP:conf/focs/Charalampopoulos20} were obtained via analogous combinatorial results:
either $P$ has \emph{few} approximate occurrences in~$T$, or $P$ is \emph{almost periodic}.

In the CT-matching model, it is not even clear what the right formalisation of periodicity~is. Indeed, several notions of periodicity have been suggested and studied thus far \cite{DBLP:journals/tcs/MatsuokaAIBT16,exactCTmatching}.
Here, we introduce a new notion of periodicity, namely \emph{CT-block-periodicity}, that provides strong locality guarantees.
We then exploit these guarantees to show that we either have few approximate CT-matches or the pattern is almost CT-block-periodic.
Dealing with the almost CT-block-periodic case is the main technical challenge we overcome.
The workhorse of our approach for this case is a combinatorial result (see \cref{lem:trim,lem:trim:nsv} for precise statements) that allows us to compute $\hdist_k(T[i \dd i+m) \leadsto P)$ very efficiently for almost all values of $i$ when $P$ is almost CT-block-periodic.
Essentially, we prove (under some extra assumptions) that, for equal-length strings $X = X_1 Q^r X_2$ and $Y = Y_1 Q^r  Y_2$ with $\absolute{X_1} = \absolute{Y_1}$, the value of $\hdist_k(X \leadsto Y)$ is the same for all choices of $r \geq 2k+1$.\footnote{We use $Q^r$ for simplicity in this overview, noting that standard string periodicity implies CT-block-periodicity. \cref{lem:trim,lem:trim:nsv} are more general.}
This allows us to trim aligned pairs of long CT-block-periodic fragments in $T[i \dd i+m)$ and $P$, thus obtaining short strings $F$ and $P'$ such that $\hdist_k(T[i \dd i+m) \leadsto P) = \hdist_k(F \leadsto P')$, and compute the latter value in $\cO(\absolute{F} \cdot k)$ time.
In summary, in the case when the number of (candidate) approximate CT-matches may be large, the periodic structure allows us to verify almost all of them efficiently.

\subparagraph{Other related work.}
Cartesian trees have received significant attention in part due to their relation to the range minimum query problem~\cite{DBLP:journals/jal/BenderFPSS05,DBLP:journals/algorithmica/DemaineLW14,DBLP:journals/siamcomp/FischerH11}.
Other problems that have been studied in the Cartesian tree matching model include computing the longest common subsequence~\cite{DBLP:conf/iwoca/TsujimotoSMNI24}, performing
subsequence matching~\cite{DBLP:conf/cpm/OizumiKMIA22}, computing palindromes~\cite{DBLP:journals/iandc/MienoFNIBT25}, and matching indeterminate strings~\cite{DBLP:conf/cpm/GawrychowskiGL20}.
Other works on strings under substring consistent equivalence relations include
\cite{DBLP:conf/walcom/Hendrian20,DBLP:conf/cpm/JargalsaikhanHY22,DBLP:conf/spire/KikuchiHYS20}.

\section{Preliminaries}

For integers $i,j \in \mathbb Z$, we write $[i, j] = [i, j + 1) = (i - 1, j] = (i - 1, j + 1)$ to denote $\{ h \in \mathbb Z \mid i \leq h \leq j\}$.
For a set $Z \subseteq \mathbb Z$ and an integer $y \in \mathbb Z$, we define $Z-y = \{z-y: z \in Z\}$.

A string $X$ of length $n = \absolute{X}$ is a finite sequence of $n$ characters
over a finite alphabet~$\Sigma$.
For $i, j \in [1,n]$, the $i$-th character of $X$ is denoted by $X[i]$, and we use either of $X[i\dd j]$, $X[i\dd j+1)$, $X(i-1\dd j]$, and $X(i-1\dd j+1)$ to denote the \emph{substring} $X[i]X[i+1]\cdots X[j]$ of $X$.
If $i > j$, then $X[i\dd j]$ is the empty string.
We say that a string $Y$ has an occurrence at position $i$ of $X$ if $i + \absolute{Y} \leq \absolute{X} + 1$ and $X[i\dd i + \absolute{Y}) = Y$.
We may refer to a substring $X[i\dd j]$ as a \emph{fragment} if we want to emphasize that we mean the specific occurrence of $X[i\dd j]$ at position $i$. Two fragments are disjoint if they do not share any positions.

An integer $p>0$ is a \emph{period} of $X$ if $X[i] = X[i + p]$ for all $i \in [1,\absolute{X}-p]$.
The smallest period of $X$ is referred to as \emph{the period} of $X$ and is denoted by $\per(X)$.
A \emph{run} in $X$ is a fragment $Y=X[i \dd j]$ that satisfies $p:=\per(Y) \le \absolute{Y}/2$ and is inclusion-maximal; that is, $X[i-1] \ne X[i-1+p]$ (or $i = 1$) and $X[j+1] \ne X[j+1-p]$ (or $j=\absolute{X}$).

\begin{lemma}[Periodicity Lemma (weak version)~\cite{fine1965uniqueness}]\label{lem:perlemma} 
If a string $S$ has periods $p$ and $q$ such that $p+q \leq \absolute{S}$, then $\gcd(p, q)$ is also a period of $S$.
\end{lemma}

Throughout the paper, we assume that the alphabet $\Sigma$ is totally ordered, and that an order comparison takes constant time.
For a string $X$, we define its \emph{leftmost minimum} as the (positioned) character $X[j]$, where $j = \min\{\, i \in [1,\absolute{X}] : X[i] = \min_{k \in [1,\absolute{X}]} X[k] \,\}$.
In other words, $j$ is the leftmost position at which the minimum character in $X$ occurs.

\subsection{Cartesian Trees}

In this work, consistent with the literature, we ensure that Cartesian trees are uniquely defined by resolving ties in favour of leftmost positions.

\begin{definition}
The Cartesian tree of a string $X$, denoted by $\CT(X)$, is a rooted ordered binary tree with $\absolute{X}$ nodes.
The Cartesian tree of an empty string is the order-zero graph.
For a non-empty string $X$, $\CT(X)$ is the tree obtained by creating a root node and (recursively) attaching to it $\CT(X[1 \dd j))$ as a left subtree and $\CT(X(j \dd \absolute{X}])$ as a right subtree, where $X[j]$ is the leftmost minimum of $X$.
\end{definition}

\begin{definition}
We say that a string $X$ CT-matches a string $Y$ (or that $X$ and $Y$ CT-match) if and only if $\CT(X)=\CT(Y)$.
We denote that two strings $X$ and $Y$ CT-match by writing $X \approx Y$ and that they do not CT-match by writing $X \not\approx Y$.
\end{definition}

\begin{fact}\label{fact:ct_transitivity}
The CT-equivalence relation ($\approx$) is a transitive relation. That is, for any three strings $X$, $Y$, and $Z$, if $X \approx Y$ and $Y \approx Z$, then $X \approx Z$ as well.
\end{fact}

For strings $T[1\dd n]$ (the text) and $P[1\dd m]$ (the pattern), we say that $i \in [1,n - m + 1]$ is a CT-occurrence of $P$ in $T$ (or that $P$ CT-occurs in $T$ at position $i$) if $P\approx T[i\dd i+\absolute{P})$. The CT-matching problem consists in computing all CT-occurrences of $P$ in $T$.
\begin{lemma}[\cite{exactCTmatching}]\label{lem:exact}
Given a text $T[1\dd n]$ and a pattern $P[1 \dd m]$, all CT-occurrences of $P$ in $T$ can be computed in $\cO(n + m)$ time.
\end{lemma}
\begin{lemma}[\cite{exactCTmatching}]\label{lem:multiple}
Given a text $T[1\dd n]$ and patterns $P_1[1 \dd m_1], \cdots, P_h[1\dd m_h]$, all the CT-occurrences of all the patterns in $T$ can be computed in $\cO((n + \sum_{i = 1}^h m_i)\log h)$~time.
\end{lemma}

For equal-length strings $X$ and $Y$, $\hdist(X \leadsto Y)$
denotes the minimum number of substitutions needed to transform $X$ into a string $X'$ such that $X' \approx Y$.
We further define $\hdist_k(X \leadsto Y) = \min \{k+1, \hdist(X \leadsto Y)\}$.
Given a text $T[1\dd n]$ and a pattern $P[1\dd m]$, we say that $P$ has a $k$-approximate CT-occurrence at a position $i \in [1,n - m + 1]$ of $T$ if $\hdist(T[i \dd i+m) \leadsto P) \leq k$.

The following fact states that CT-matching is a substring consistent equivalence relation.

\begin{fact}[Theorem 2, \cite{exactCTmatching}]\label{fact:subeq}
For any strings $U$ and $V$ with $U \approx V$, we have $U[i \dd j] \approx V[i \dd j]$ for all $i, j \in [1,\absolute{U}]$.
\end{fact}

There is a natural bijection between positions of a string $X$ and nodes in $\CT(X)$.
We formalise this concept using the previous-smaller-or-equal-value function.
For a string $X$, we define
\begin{align*}
\psv_X : \quad & [1, \absolute{X}] \rightarrow [0, \absolute{X}), 
        & i \mapsto \max(\{j \in [1, i) \mid X[j] \leq X[i]\} \cup \{0\}).
\end{align*}
The parent-distance representation $\PD(X)$ of a string $X$ (see \cite{exactCTmatching}) is a length-$\absolute{X}$ string with
\[
\PD(X)[i]= 
\begin{cases}
   i - \psv_X(i)& \text{ if } \psv_X(i) \neq 0;\\
   0 &\text{otherwise.}
\end{cases}
\]
Intuitively, we have $\psv_X(i)=j$ if and only if
the left subtree of the node corresponding to $X[i]$ in $\CT(X)$
is $\CT(X(j \dd i))$.
Given either of $\psv_X$ and $\PD(X)$, one can (recursively) construct $\CT(X)$, and hence the following lemma holds.

\begin{lemma}[{\cite[Theorem 1]{exactCTmatching}}]\label{lem:psv_equiv}
    For equal-length strings $X$ and $Y$, $X \approx Y$ if and only if $\PD(X) = \PD(Y)$, or equivalently $\forall i \in [1, \absolute{X}] : \psv_X(i) = \psv_Y(i)$.
\end{lemma}

We provide an example of a Cartesian tree $\CT(X)$, the function $\psv_X$, the array $\PD(X)$, and the distance function $\hdist$ in \cref{fig:ct-tree}.
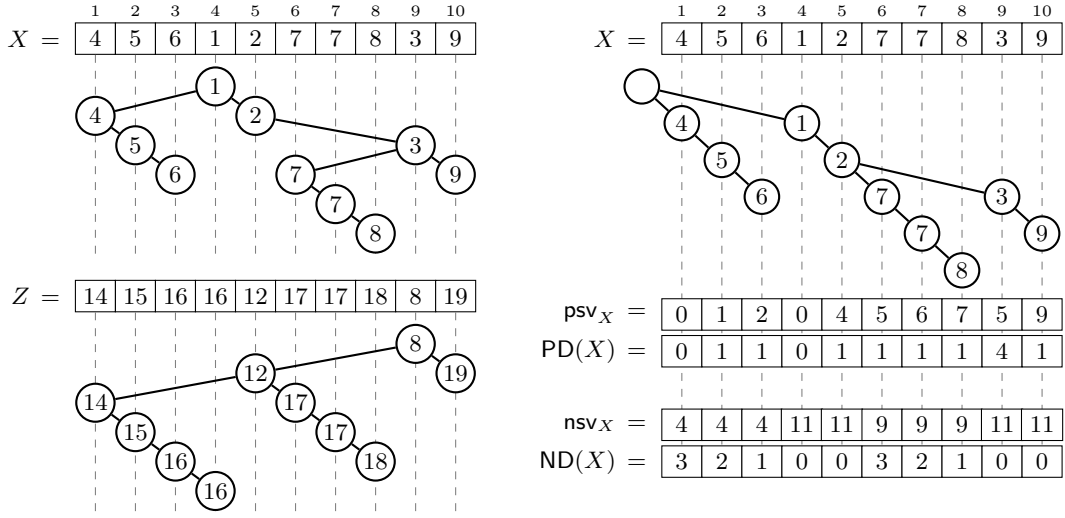
\begin{figure}
\centering
{\input{figs/CT-tree-tables}}
\caption{The Cartesian tree $\CT(X)$ of
$X = [4, 5, 6, 1, 2, 7, 7, 8, 3, 9]$ (top left), the tree induced by $\psv_X$ (top right), and the values of $\psv_X$, $\nsv_X$, $\PD(X)$, and $\ND(X)$ (bottom right). For $Y = [14, 15, 16, 11, 12, 17, 17, 18, 13, 19]$, we have $X \approx Y$, illustrating that different sequences can yield the same Cartesian tree. In contrast, the Cartesian tree of $
Z = [14, 15, 16, 16, 12, 17, 17, 18, 8, 19]$ (bottom left) is different, i.e., we have $Z \not\approx X$. However, we have $\hdist(X \leadsto Z) = 2$ because we can substitute the values of $X$ at positions 4 and 9 to obtain $
X' = [4, 5, 6, \mathbf{6}, 2, 7, 7, 8, \mathbf{1}, 9]\approx Z$.}
\label{fig:ct-tree}
\end{figure}%
For a string $X$, analogously to $\psv_X$, we define the
next-smaller-value function, which captures the structure of the Cartesian tree as well:
\begin{align*}
\nsv_X : \quad & [1, \absolute{X}] \to [2, \absolute{X}+1], 
        & i \mapsto \min\bigl(\{\, j \in (i, \absolute{X}] \mid X[j] < X[i] \,\} \cup \{\absolute{X}+1\}\bigr).
\end{align*}%
We also define a length-$\absolute{X}$ string $\ND(X)$ by
\[
\ND(X)[i] =
\begin{cases}
\nsv_X(i) - i, & \text{if } \nsv_X(i) \neq \absolute{X}+1, \\
0, & \text{otherwise}. 
\end{cases}
\]
An analogue of \cref{lem:psv_equiv} for $\nsv_X$ and $\ND(X)$ is provided
as \cref{lem:nsv_equiv} in \cref{app:symmetry}.

\section{Cartesian Periodicity}

Several definitions of periodicity for strings with respect to Cartesian tree matching have been studied (see~\cite{DBLP:journals/tcs/MatsuokaAIBT16,exactCTmatching}).
Here, we use the following definitions.

\begin{definition}[\cite{DBLP:journals/tcs/MatsuokaAIBT16}]
    A string $P[1 \dd m]$ has a CT-border-period $p \in [1, m]$ if and only if $P[1 \dd m - p] \approx P(p \dd m]$.
\end{definition}

\begin{lemma}\label{lem:period:minimum_in_border}
    The leftmost minimum of a string $P[1 \dd m]$ with CT-border-period $p \in [1, \floor{\frac m2}]$ lies outside $P(p \dd m-p]$.
\end{lemma}

\begin{proof}
    Assume that $i \in (p, m-p]$ is the position of the leftmost minimum of $P$.
    Then $P[i]$ is also the leftmost minimum of both $P[1 \dd m-p]$ and $P(p \dd m]$, and hence it is the root of both $\CT(P[1 \dd m-p])$ and $\CT(P(p \dd m])$. However, in this case, the left subtree under the root of $\CT(P[1 \dd m-p])$ is of size $i - 1$, while the the left subtree under the root of $\CT(P(p \dd m])$ is of size $i - p - 1$.
    This contradicts the assumption that $P[1 \dd m-p] \approx P(p \dd m]$.
\end{proof}

\begin{definition}
    A string $P[1 \dd m]$ has a CT-block-period  $p \in [1, m]$ if and only if
    \begin{itemize}
        \item $p$ divides $m$,
        \item $p$ is a CT-border-period of $P$, and
        \item the leftmost minimum of $P$ is either $P[1]$ or $P[m]$.
    \end{itemize}
\end{definition}

\begin{remark}
The above definition restricts the definition of a \emph{full period} from \cite{PARK2020181} by imposing the third condition.
This condition crucially ensures that, for a CT-block-periodic fragment $X[i\dd j]$ of a string $X$, either $\psv_X[k] \in [i,j)$ for all $k \in (i,j]$ or $\nsv_X[k] \in (i,j]$ for all $k \in [i,j)$, thus providing locality guarantees.
\end{remark}

\begin{lemma} \label{lem:period:border_to_block}
    Given $P[1 \dd m]$ and a positive integer $p < m/3$ that is a CT-border-period of~$P$, we can compute a fragment of length at least $m - 2p$ that has CT-block-period $p$ in $\cO(p)$ time.
\end{lemma}

\begin{proof}
    \cref{fact:subeq} implies that any substring of $P$ of length at least $p$ has CT-border-period~$p$.
    Let $m' = p \cdot \ceil{\frac{m}{p}} - 2p$ and note that $p < m - 2p \leq m' < m - p$.
    By \cref{lem:period:minimum_in_border}, there exists an offset $a \in [0, p)$ such that the leftmost minimum of $P$ is either $P[1 + a]$ or $P[m - a]$.
    If the leftmost minimum is $P[1 + a]$, then the fragment $P[1 + a \dd m' + a]$ has CT-block-period~$p$. 
   Otherwise, the fragment $P(m - a - m' \dd m - a]$ has CT-block-period~$p$.
\end{proof}

\begin{lemma}\label{lem:blockper_chain}
    Consider a string $P[1 \dd m]$ that has a CT-block-period $p \in [1, \floor{\frac m2})$ and such that $P[1]$ is the leftmost minimum of~$P[1 \dd m-p]$.
    For $i \in [1, m/p]$, let $c_i = (i-1) \cdot p + 1$.
    For all $i \in [1, m/p]$, $P[c_i]$ is the leftmost minimum of $P[c_i \dd m]$.
\end{lemma}
\begin{proof}
    For $i \in [1,m/p]$ denote $P_i := P[c_i \dd m]$.
    We prove by induction the more general statement that $p$ is a CT-block-period of $P_{i}$ for all $i$ with $P_{i}[1]$ being the leftmost minimum of~$P_{i}$.
    The induction hypothesis holds for $i=1$ since $P[1]$ is the leftmost minimum of $P[1 \dd m-p]$ and is therefore smaller than $P[p+1]$ (since $p<m/2$), where  $P[p+1]$ is the leftmost minimum of $P[p+1 \dd m]$.
    Now, consider any $i \in [2,m/p]$ and suppose that $p$ is a CT-block-period of $P_{i-1}$ with $P_{i-1}[1]$ being the leftmost minimum of~$P_{i-1}$.
    First, since $p$ is a CT-block-period of $P_{i-1}$, we have $P_{i-1}[1 \dd \absolute{P_i}] \approx P_i$ and hence the leftmost minimum of~$P_i$ is $P_i[1]$.
    Further, \cref{fact:subeq} implies that
    \[P[c_i \dd m - p] = P_{i-1}[1 + p\dd \absolute{P_{i-1}} - p] \approx 
    P_{i-1}[1+2p\dd \absolute{P_{i-1}}] = P[c_i + p \dd m]\]
    and hence $p$ is a CT-border-period of $P_{i}$.
    Finally, since $p$ divides $\absolute{P}$, it also divides $\absolute{P_i} = \absolute{P} - (i-1)p$.
    This completes the induction and the statement follows.
\end{proof}

An analogue of \cref{lem:blockper_chain} for the case when $P[m]$ is the leftmost minimum of~$P$ is given
as \cref{lem:blockper_chain_v2} in \cref{app:symmetry}.

\begin{fact}\label{fact:pd_fragment}
Consider a string $X$.
The parent-distance representation 
of a fragment $X[i \dd j]$ of $X$ is given by, for $t \in [1, j - i + 1]$,
\[
\PD(X[i \dd j])[t] =
\begin{cases}
  0, & \text{if } \PD(X)[i+t-1] \geq t, \\[6pt]
  \PD(X)[i+t-1], & \text{otherwise}.
\end{cases}
\]
If $X[i]$ is a leftmost minimum of $X[i \dd j]$, then $\PD(X)(i\dd j]$ is a suffix of $\PD(X[i \dd j])$.
\end{fact}

\begin{lemma}\label{lem:no_overlaps}
Consider a string $X$ and a fragment $X[i \dd j]$ with leftmost minimum $X[i]$.
Then, $p < j - i + 1$ is a CT-block-period of $X[i\dd j]$ if and only if $p$ is a period of $\PD(X)(i\dd j]$ and $p$ divides $j-i+1$.
\end{lemma}
\begin{proof}
($\Rightarrow$:) By the definition of a CT-block-period, we have that $p$ divides $j-i+1$ and that
$X[i \dd j-p] \approx X[i+p \dd j]$.
By \cref{lem:psv_equiv}, this equality implies that
\begin{equation}\label{eq:simple}
\PD(X[i \dd j-p]) = \PD(X[i+p \dd j]).
\end{equation}
Further, by \cref{fact:pd_fragment}, since $X[i]$ and $X[i+p]$ are the leftmost minima of $X[i \dd j-p]$ and $X[i+p \dd j]$, respectively, we have that
$\PD(X)(i \dd j-p]$ is a suffix of $\PD(X[i \dd j-p])$ and that
$\PD(X)(i+p \dd j]$ is a suffix of $\PD(X[i+p \dd j])$.
By \eqref{eq:simple}, we then have $\PD(X)(i \dd j-p] = \PD(X)(i+p \dd j]$
and hence $p$ is a period of $\PD(X)(i \dd j]$.

($\Leftarrow$:)
We first argue that it suffices to show that $X[i+p]$ is the leftmost minimum of $X[i+p \dd j]$.
If this is indeed the case, we have that $\PD(X)(i+p\dd j]=\PD(X)(i \dd j-p]$ is a suffix of $\PD(X[i+p \dd j])$ by \cref{fact:pd_fragment}. This, together with the fact that $\PD(X[i \dd j-p])[1] = 0 = \PD(X[i+p \dd j])[1]$, implies that $\PD(X[i \dd j-p]) = \PD(X[i+p \dd j])$.
We hence have that $X[i \dd j-p] \approx X[i+p \dd j]$ by \cref{lem:psv_equiv}.
Since $p$ divides $j-i+1$ and $X[i]$ is the leftmost minimum of $X[i\dd j]$, $p$ is a CT-block-period of $X[i\dd j]$.

Now, suppose for the sake of contradiction, that the leftmost minimum of $X[i+p \dd j]$ is at some position $z \in (i+p, j]$.  
Then, $\psv_X(z) = z - \PD(X)[z] < i+p$.
Since $X[i]$ is the leftmost minimum of $X[i \dd j]$, $\psv_X(z) \in [i,i+p)$.
Since $\PD(X(i \dd j])$ is periodic with period $p$, we have 
$\PD(X)[z] = \PD(X)[z-p]$.
Therefore, our assumption implies that $\psv_X(z-p) = (z-p) - \PD(X)[z-p] = \psv_X(z) - p < i$,
which contradicts the assumption that $X[i]$ is the leftmost minimum of $X[i \dd j]$.
\end{proof}

An analogue of the above lemma for the case when the leftmost minimum of $X[i \dd j]$ is $X[j]$ is given
as \cref{lem:no_overlaps_v2} in \cref{app:symmetry}.

\begin{corollary}\label{cor:focusonp}
For equal-length strings $X$ and $Y$ that have a common CT-block-period $p \in [1, \floor{\absolute{X}/2}]$, we have $X \approx Y$ if and only if $X[1 \dd 2p] \approx Y[1 \dd 2p]$.
\end{corollary}
\begin{proof}
($\Rightarrow$:) This implication follows directly from \cref{fact:subeq}.

($\Leftarrow$:)
Due to \cref{lem:blockper_chain} and its analogue (\cref{lem:blockper_chain_v2}), both $X$ and $Y$ have their leftmost minima at the same position (since $X[1 \dd 2p]$ and $Y[1 \dd 2p]$ do).
For the remainder of the proof, we assume that these minima occur at the first positions, noting that the remaining case can be handled symmetrically.
Since $X[1 \dd 2p] \approx Y[1 \dd 2p]$, \cref{lem:psv_equiv} yields that $\PD(X[1 \dd 2p]) = \PD(Y[1 \dd 2p])$.
Further, we have $\PD(X[1 \dd 2p])=\PD(X)[1 \dd 2p]$ and $\PD(Y[1 \dd 2p])=\PD(Y)[1 \dd 2p]$ by \cref{fact:pd_fragment}.
Hence, $\PD(Y)[1 \dd 2p]=\PD(X)[1\dd 2p]$.
Since $p$ is a CT-block-period of each of $X$ and $Y$, by applying \cref{lem:no_overlaps} to each of $X$ and $Y$, we conclude that $p$ is a period of both $\PD(X)(1 \dd \absolute{X}]$ and $\PD(Y)(1 \dd \absolute{Y}]$.
This directly implies that $\PD(X)=\PD(Y)$, and we hence have $X \approx Y$ by \cref{lem:psv_equiv}. 
\end{proof}

\begin{lemma}\label{lem:per_frag}
Suppose that a fragment $X[i \dd j)$ has minimum CT-block-period $p$. Then, for any $a,b \in [i,j]$ such that $a \equiv i \pmod p$, $b \equiv j \pmod p$, and $b-a \geq 2p$, we have that $X[a \dd b) \approx X[i \dd i + (b-a))$ and that $X[a \dd b)$ has minimum CT-block-period $p$.
\end{lemma}

\begin{proof}
    Without loss of generality, assume that $X[i]$ is the leftmost minimum of $X[i\dd j)$.
    The complementary case when $X[j-1]$ is the leftmost minimum can be handled in a symmetric manner using the analogues of the lemmas used below.
    
    Fix $a,b$ satisfying the conditions of the claim. If $a \neq i$, then the CT-border-period $p$ of $X[i\dd j)$ implies that
    $X[a \dd b) \approx X[a-p \dd b-p) \approx \cdots \approx X[i \dd i+(b-a)).$

    Now, observe that $X[a]$ is the leftmost minimum of $X[a\dd j)$ by \cref{lem:blockper_chain}.
    Additionally, we have $X[a \dd b-p)\approx X[a+p\dd b)$ by \cref{fact:subeq} because $p$ is a CT-border-period of $X[i\dd j)$.
    Finally, $p$ divides $b-a$ by our assumptions and hence $p$ is a CT-block-period of $X[a \dd b)$.
    
    It remains to show that $p$ is the minimum CT-block-period of $X[a\dd b)$.
    Towards a contradiction suppose that $X[a\dd b)$ has minimum CT-block-period $q<p$.
    Then, by \cref{lem:no_overlaps}, $\PD(X)(a\dd b)$ has a period $q$ and $\PD(X)(i \dd j)$ has a period $p$.
    Since the length of the fragment is $b-a \geq 2p > p+q$, an application of the Periodicity Lemma (\cref{lem:perlemma}) yields that $\PD(X)(a\dd b)$ has period $\gcd(p,q)$. Consequently, $\PD(X)(i \dd j)$ also has period $\gcd(p,q)$.
    Another application of \cref{lem:no_overlaps} implies that $X[i \dd j)$ has CT-block-period $\gcd(p,q)<p$ (since $\gcd(p,q)$ divides the length $j-i$), a contradiction.
\end{proof}

\begin{definition}\label{def:CTrun}
    Let $X$ be a string.
    A fragment $X[i \dd j]$  with minimum CT-block-period $p \leq (j-i+1)/2$
    is called a \emph{CT-run}
    if and only if
    $p$ is not a CT-block-period of
    $X[i-p \dd j]$ (or $i-p < 1$)
    and $p$ is not a CT-block-period of $X[i \dd j+p]$ (or $j+p > \absolute{X}$)\color{black}.
\end{definition}

    The occurrences of a periodic pattern $P$ in a text $T$ form arithmetic progressions with each such progression corresponding to a run with period $\per(P)$.
    We show an analogous result for CT-occurrences.

\begin{example}\label{ex:ct-match}
Consider a pattern $P = [10, 40, 30, 20, 60, 50]$.
$P$ is of length $6$ and has a CT-block-period $p:=3$. 
The leftmost minimum of $P$ is at its first position.
The parent-distance array for $P$ is
$\PD(P) = [0, 1, 2, 3, 1, 2]$.

Consider the text $T = [100, 400, 300, 200, 600, 500, 300, 800, 700, 900]$ of length $n=10$. 
The parent-distance array for $T$ is: $\PD(T) = [0, 1, 2, 3, 1, 2, 3, 1, 2,1]$. 
\smallskip

There are two CT-occurrences of $P$ in $T$:
\begin{itemize}
    \item at position $1$, $T[1 \dd 6] = [100, 400, 300, 200, 600, 500]$, where $\PD(T[1 \dd 6]) = [0, 1, 2, 3, 1, 2]$;
    \item at position $4$, $T[4 \dd 9] = [200, 600, 500, 300, 800, 700]$, where $\PD(T[4 \dd 9]) = [0, 1, 2, 3, 1, 2]$.
\end{itemize}

\noindent Note that the fragment $T[1\dd 9]$ forms a CT-run in $T$ with minimum CT-block-period $p=3$.
\end{example}

\begin{lemma}\label{lem:not_folklore}
    Let $P$ be a pattern of length $m$ with minimum CT-block-period $p\leq m/4$, and let $T$ be a text of length $n$.
    The exact CT-occurrences of $P$ in $T$ can be partitioned into $\cO(n/m)$ maximal arithmetic progressions with common difference $p$.
    For each such  progression $\{u + vp : v \in [0,z]\}$, the fragment $T[u \dd u + zp +m)$ is a CT-run with minimum CT-block-period $p$.
    For any two such arithmetic progressions $\{u_1 + vp : v \in [0,z_1]\}$ and $\{u_2 + vp : v \in [0,z_2]\}$, the two runs $T[u_1 \dd u_1 + z_1p +m)$ and $T[u_2 \dd u_2 + z_2p +m)$ overlap by fewer than $p$ positions.
\end{lemma}

\begin{proof}
    We henceforth assume that the leftmost minimum of $P$ is at its first position. The proof for the complementary case, that is, when the leftmost minimum of $P$ is $P[m]$, is symmetric (using $\ND$ instead of $\PD$ and the analogous versions of the used lemmas).

    Let $P' := \PD(P)(1\dd m]$. By the CT-block-periodicity of $P$ and 
    \cref{lem:no_overlaps} we have $\per(P') = p$.
    Since $P[1]$ is the leftmost minimum of $P$, \cref{fact:pd_fragment,lem:psv_equiv} imply that $P$ CT-occurs at some position $x$ in $T$ if and only if the following two conditions hold:
    \smallskip
    \begin{itemize}
        \item \emph{$\PD$-match condition}: $\PD(T)(x\dd x+m) = P'$;\label{it:condone}
        \item \emph{left-min condition}: $T[x]$ is the leftmost minimum of $T[x\dd x+m)$.\label{it:condtwo}
    \end{itemize}

\begin{claim}\label{claim:experiment}
    Consider a maximal arithmetic progression $D = \{j + i p : i\in [0,s]\}$ of occurrences of $P'$ in $\PD(T)$. For $i\in [0,s]$, denote $d_i := j + i p$. Then, either $P$ does not have any CT-occurrence at a position of $T$ in $D-1$, or there exists an integer $\mu \in [0,s]$ such that both of the following hold:
    \smallskip
    \begin{itemize}
        \item the intersection of the set of CT-occurrences of $P$ in $T$ and $D-1$ is $\{d_i - 1 : i\in [\mu,s]\}$;
        \item $T[d_\mu - 1 \dd d_s-1+m)$ is a CT-run with minimum CT-block-period $p$.
    \end{itemize}
    \end{claim}   
    \begin{claimproof}
        By the definition of $D$, we have $\PD(T)(x\dd x+m)=P'$ for all $x \in D$. Hence, the
        $\PD$-match condition holds for all $x \in D-1$. If
        the left-min condition fails for every element of $D-1$, then $P$ does not have any CT-occurrence at any position of~$T$ in~$D-1$, and the claim follows.
        Otherwise, let $\mu = \min \{i \in [0,s] : T[d_i - 1] \mbox{ is a minimum in } T[{d_i-1} \dd {d_i-1+m})\}$.
        We show  by induction that
        the left-min condition holds for all $x \in \{d_i - 1 : i\in [\mu,s]\}$. The base case $y=\mu$ holds by definition. Assume that it holds for some $y \in [\mu,s)$. An application of  \cref{lem:blockper_chain} to $T[d_y - 1\dd d_{y+1}-1+m)$  implies that
        $T[d_{y+1} - 1]$ is the leftmost minimum of
        $T[d_{y+1}-1 \dd d_{y+1}-1+m)$.
        This completes the proof of the inductive step and hence concludes the proof of the first item.
        Let $F := T[d_{\mu} - 1 \dd d_s-1+m)$.
        By repeatedly applying \cref{lem:blockper_chain}, we conclude that the leftmost minimum of $F$ is $F[1]$, and hence that~$p$ is a CT-block-period of $F$.
        If $F$ had a smaller CT-block-period, then, by \cref{lem:per_frag}, the same would hold for $P$, contradicting the assumption that $P$ has minimum CT-block-period~$p$.
        Finally, we cannot prepend or append $p$ characters to $F$ while maintaining that it has CT-block-period $p$: either the obtained string does not have its leftmost minimum at its first position or we contradict the maximality of~$D$. This concludes the proof of the claim.
    \end{claimproof}
        
    To complete the proof, it suffices to combine \cref{claim:experiment} with the fact that a CT-occurrence of $P$ in $T$ at position $j$ implies an occurrence of $P'$ in $\PD(T)$ at position $j+1$, noting that the occurrences of $P'$ in $\PD(T)$ can be partitioned into $\cO(n/m)$ maximal arithmetic progressions with period $\per(P') = p$ such that any two elements from distinct arithmetic progressions differ by more than $m-p$ (see \cite{BG95,fine1965uniqueness}).
\end{proof}

\section{Approximate CT-Matching with Substitutions}

Kim and Han \cite{DBLP:journals/tcs/KimH25} recently presented a dynamic programming algorithm that, given two strings $X$ and $Y$, computes the minimum cost of a sequence of edit operations (insertions, deletions, and substitutions) required to transform $X$ into a string $X'$ such that $X'\approx Y$.
As noted in \cite[Section 6]{DBLP:journals/tcs/KimH25}, a minor modification to their algorithm results in the following lemma for the special case when only substitutions are allowed.

\begin{restatable}[see \cite{DBLP:journals/tcs/KimH25}]{lemma}{compare}\label{lem:compare}
    Given two strings $T[1\dd m]$ and $P[1\dd m]$, and an integer threshold $k>0$, we can compute $\hdist_k(T \leadsto P)$ in $\cO(mk)$ time.
\end{restatable}

\noindent%
In \cref{app:compare}, we provide a self-contained proof of \cref{lem:compare} for the sake of completeness.

\subsection{Analysing the Pattern}

\begin{definition}
A set $\mathcal{S}$ of equal-length strings is a
CT-rainbow if and only if there do not exist distinct $S_1, S_2 \in \mathcal{S}$ such that $S_1 \approx S_2$.
\end{definition}

\begin{lemma}\label{lem:filtering_preprocess}
Let $m, \Delta \in \mathbb Z_+$ with $\Delta < \sqrt{m} / 2$.
Given $\Delta$ and a string $P[1 \dd m]$, in $\cO(m\Delta)$ time, we can either compute a set of $\Delta$ disjoint fragments of $P$, each of length $\ceil{m/{(2\Delta)}}$, that form a CT-rainbow or report a fragment of length at least $m/{(4\Delta)}$ that has a CT-block-period less than~$\Delta$, together with its minimal CT-block-period.
\end{lemma}

\begin{proof}
    Let $m' = \ceil{m / (2\Delta)}$.
    We process $P$ in a left-to-right manner, attempting to identify $m'$-length fragments to insert into an initially empty set $\mathcal{S}$,
    which will remain a CT-rainbow throughout.
    Whenever a fragment is inserted into $\mathcal{S}$, we find and mark all of its exact CT-occurrences in $P$ in $\cO(m)$ time using \cref{lem:exact}.
    For each inserted fragment, we use a~unique symbol for marking.

    We first insert fragment $P[1 \dd m']$ into $\mathcal{S}$.
    Let $P[i \dd i+m')$ be the fragment that was most recently inserted into $\mathcal{S}$ and assume that $\absolute{\mathcal{S}} < \Delta$ (otherwise, we return $\mathcal{S}$ and terminate the algorithm).
    Our goal is to then find some $j \in [i + m', i+m' + \Delta)$ such that $P[j \dd j+m')$ does not CT-match any element of $\mathcal{S}$.
    To do this, it suffices to find, in $\cO(\Delta)$ time, the first unmarked position in $[i + m', i+m' + \Delta)$, if one exists.
    If we can always find such a position~$j$ while $\absolute{\mathcal{S}} < \Delta$, then we eventually insert $\Delta$ fragments into $\mathcal{S}$.
    This holds because the total length of the fragments inserted into $\mathcal{S}$ and the lengths of the gaps between them add up to at most $\Delta \cdot (m' + \Delta - 1) \leq {m/2} + \Delta^2 \leq m$ since the length of the gap between any two consecutive fragments is less than $\Delta$ and $\Delta < \sqrt{m}/2$.
    Inserting a fragment into $\mathcal{S}$ and marking all of its exact CT-occurrences takes $\cO(m)$ time.
    In the case when we can always find a fragment to insert into $\mathcal{S}$, the algorithm thus runs in $\cO(m\Delta)$ time.

    Now consider the complementary case when, after inserting some fragment $P[i \dd i+m')$ into $\mathcal{S}$, we have $\absolute{\mathcal{S}}< \Delta$ and all positions in $[i + m', i+m' + \Delta)$ are marked.
    Then, by the pigeonhole principle, some fragment in $\mathcal{S}$ has two CT-occurrences at positions $i_1 < i_2$ in $[i + m', i+m' + \Delta)$ with $q := i_2 - i_1 < \Delta$.
    We find two such positions (marked with the same symbol) naively in $\cO(\Delta^2) \subseteq \cO(m)$ time. Let $G= P[i_1 \dd i_2 + m')$.
    $G$ has length $L = m' + q$ and CT-border-period $q$.
    We apply \cref{lem:period:border_to_block} to trim $G$. Both conditions of the lemma are satisfied since $\Delta > q$, $m' \geq m / (2\Delta)$, and $\Delta < \sqrt{m}/2$, and hence
    \[L = m' + q \geq m / (2\Delta) + q > \sqrt{m} + q \geq 2\Delta + q > 3q.\]
    Trimming $G$ yields a fragment $F$ with CT-block-period $q$ and length \[\absolute{F} \geq L-2q = m'-q \geq {m / (2\Delta)} - \Delta > {m / (2\Delta)} - \sqrt{m} / 2 \geq m / (4\Delta).\]
    Finally, we naively compute the smallest CT-block-period $p$ of $F$ by checking each integer in $[1,q-1]$ in $\cO(m)$ time, for a total running time of $\cO(m\Delta)$, and return the pair $(F,p)$.
\end{proof}

\subsection{The Aperiodic Case}

In this section, we deal with the (easy) case when the analysis of the pattern according to \cref{lem:filtering_preprocess} returns a large CT-rainbow of disjoint fragments.
We apply a standard marking trick: an exact CT-occurrence of a fragment $P[a \dd b] \in \mathcal{S}$ at a position $j$ in $T$ contributes a mark to position $j-a+1$, indicating a potential starting position of a $k$-approximate CT-occurrence of $P$ in~$T$.
The following simple fact captures the key property we use: if ${\hdist_k(T[i \dd i+m) \leadsto P) \leq k}$, then position $i$ must receive at least $\Delta - k$ marks.

\begin{lemma}\label{lem:comb_mark}
   Let $k, \Delta \in \mathbb Z_{+}$ with $\Delta \geq k$.
   Consider strings $V,U$ satisfying ${\hdist_k(V \leadsto U)\leq k}$.
   Let $\mathcal{U} = \{U[a_i \dd b_i] : i \in [1, \Delta]\}$ be a set of disjoint fragments of $U$. Then there exists a set $I \subseteq [1,\Delta]$ of size at least $\Delta - k$ such that $V[a_i \dd b_i] \approx  U[a_i \dd b_i]$ for all $i \in I$.
\end{lemma}
\begin{proof}
    Let $V'$ be a string obtained from $V$ by substituting at most $k$ characters such that $V' \approx U$.
    Let $I := \{i \in [1,\Delta] : V[a_i \dd b_i] = V'[a_i \dd b_i]\}$.
    Set $I$ is of size at least $\Delta-k$ as the substitutions transforming $V$ into $V'$ affect at most $k$ fragments $V[a_i \dd b_i]$.
    For each $i \in I$, \cref{fact:subeq} implies that $V'[a_i \dd b_i] \approx U[a_i \dd b_i]$.
    Moreover, for each $i\in I$, we have $V[a_i \dd b_i] = V'[a_i \dd b_i]$ and hence the claim follows.
\end{proof}

In the following lemma, we use the \emph{marking trick} to efficiently compute a small number of candidate positions that need to be verified, and verify them using \cref{lem:compare}.

\begin{lemma}\label{lem:aper_case}
    Let $n,m,k, \Delta \in \mathbb Z_+$ with $n \leq 3m/2$ and $2k \leq \Delta \leq \sqrt{m} / 2$.
    Given $k$, a text $T[1 \dd n]$, and a pattern $P[1 \dd m]$, together with a CT-rainbow $\mathcal{P}$ consisting of $\Delta$ disjoint fragments of $P$, each of length $\ceil{m/{2\Delta}}$, we can compute $\hdist_k(T[i \dd i+m) \leadsto P)$ for all $i\in [1,n-m+1]$ in $\cO(m^2k / \Delta)$ time.
\end{lemma}
\begin{proof}
    We perform multiple-pattern CT-matching in $T$ for all fragments of $\mathcal{P}$ in $\cO(m \log m) \subset \cO(m^2k / \Delta)$ time using \cref{lem:multiple}.
    If a fragment $P[a \dd b] \in \mathcal{P}$ CT-matches a fragment $T[i \dd j]$, we add a mark at position $i - a + 1$ of $T$.
    By the fact that $\mathcal{P}$ is a CT-rainbow and transitivity (\cref{fact:ct_transitivity}), at most one fragment of $\mathcal{P}$ can CT-match at any position of~$T$.
    Hence, the total number of  marks is  $\cO(m)$.
    Furthermore, \cref{lem:comb_mark} implies that if $\hdist_k(T[i \dd i+m) \leadsto P) \leq k$, then position $i$ must receive at least $\Delta - k \geq \Delta /2$ marks.
    There are $\cO(m/\Delta)$ such positions, which can be identified at no additional cost during the marking process. 
    Each of these  positions is then verified  in $\cO(mk)$ time using \cref{lem:compare}, yielding a total 
    verification time of $\cO(mk \cdot m/\Delta) = \cO(m^2k / \Delta)$.
\end{proof}

\subsection{The Periodic Case}

In this section, we address the (much more complicated) case when the analysis of the pattern according to \cref{lem:filtering_preprocess} returns a long fragment with a small CT-block-period.

We first establish combinatorial lemmas
that allow us to trim both~$P$ and $T[i \dd i+m)$ under certain conditions before applying \cref{lem:compare}.
Intuitively, whenever a long CT-block-periodic fragment of $T$ is aligned with a corresponding periodic fragment of $P$,  substitutions are not required in (a long part of) this fragment.
This trimming improves computational efficiency, as the complexity of \cref{lem:compare} depends on the lengths of the input strings.

\begin{figure}
    \centering
    \def\xl{.58em}
    \def\yl{1.1em}
    \def\nc{64}
    \newcommand{\pedge}[3][black!20!white,thin]{
        \node (src) at (n#2.260) {};
        \node (dst) at (n#3.280) {};
        \pgfmathparse{.6 + 1*((#2-#3-1)/63)^(5)}
        \edef\vdist{\pgfmathresult}
        \node[below=\vdist of src] (v) {};
        \pgfmathparse{.5 - .3*(floor((#2-#3)/40))}
        \edef\loosn{\pgfmathresult}
        \draw[-{Latex[right]},#1] (src.center) to (src |- v) to[in=270, out=270, looseness=\loosn] (dst |- v) to (dst.center);
    }
    \begin{tikzpicture}[x=\xl, y=\yl, 
    every node/.style={inner sep=0}, 
    unitnode/.style={minimum width=\xl, minimum height=\yl},
    xhalfsnug/.style={inner xsep=-.2pt},
    yhalfsnug/.style={inner ysep=-.2pt},
    every fit/.style={xhalfsnug,yhalfsnug},   
    cedge/.style={black},
    ]
        \foreach \i in {0,...,\nc} {
            \node[unitnode] (nn\i) at (\i, 2.5) {};
            \node[unitnode] (n\i) at (\i, 0) {};
        }

        \node[fit=(nn18)(nn38), pattern=north east lines, pattern color=red, draw] {};

        \foreach[count=\i from 1] \x in {18,34} {
        \draw[line width=3pt, white] (n\x.south east) to (nn\x.north east) to node[pos=1, inner sep=3pt] (dl\i) {} ++(0, 1);
        \draw[densely dashed] (n\x.south east) to (nn\x.north east) to ++(0, 1);
        }

        \draw[thick, decorate, decoration={brace, amplitude=5pt}] (dl1.south east) -- node[midway, above=1, unitnode] (brr) {} (dl2.south west);
        
        \node[fit=(nn1)(nn\nc), draw] (txtX) {};
        \node[fit=(n1)(n\nc), draw] (txt) {};
        
        \node[left=0 of txt]  {\small$Y =\enskip$};
        \node[left=0 of txtX] (txtXX) {\small$X =\enskip$};

        \node[above left=2 and 0 of txtX.north east, align=center] (txtfirstlemma) {\small \cref{lem:no_sub_in_run}: No substitutions here\\\small when achieving $\hdist(X \leadsto Y)$.};

        \node[below right=0 and -.5em of txtXX.north west |- txtfirstlemma.north, align=center, outer xsep=.5em] (txtsecondlemma) {\small \cref{lem:trim}: Trimming this fragment\\\small does not affect $\hdist(X \leadsto Y)$.};

        \draw[line width=4pt, white] (nn37.west) to ++(1,0) to[out=0, in=225] (txtfirstlemma.south west);
        \draw[Latex-] (nn37.west) to ++(1,0) to[out=0, in=225] (txtfirstlemma.south west);

        \draw[-Latex] (txtsecondlemma.east) to[out=0, in=90] (brr.south);


        \node[fit=(n1)(n5), draw] (pref) {};
        \node at (pref) {\small$Y_1$};

        \node[fit=(n54)(n64), draw] (suff) {};
        \node at (suff) {\small$Y_2$};

        \node[fit=(nn1)(nn5), draw] (npref) {};
        \node at (npref) {\small$X_1$};

        \node[fit=(nn54)(nn64), draw] (nsuff) {};
        \node at (nsuff) {\small$X_2$};

        \foreach[count=\i from -4,
        evaluate=\x as \xa using int(\x+1),
        evaluate=\x as \xb using int(\x+2),
        evaluate=\x as \xc using int(\x+3),
        ] \x in {34,38,42,46,50,6,10,14,18,22} {
            \node[fit=(n\x), draw, fill=black!20!white] (nc\i) {};
            \node[above=.2 of nc\i |- txt.north, fill=white, inner ysep=1pt, inner xsep=.2em] (ncl\i) {\small%
            \dagstuhlifnum\i=0$c_r$\else
            \dagstuhlifnum\i<0$c_{r\i}$\else
            $c_\i$\fi\fi};
            
            \pedge{\xa}{\x}            
            \pgfmathrandominteger{\cdst}{\x}{\xb}
            \dagstuhlifnum\cdst=\xb\pgfmathrandominteger{\bdst}{\x}{\xa}
            \else\pgfmathrandominteger{\bdst}{\cdst}{\xa}\fi
            \pedge{\xc}{\cdst}
            \pedge{\xb}{\bdst}
        }

        \node[fit=(n38), draw, fill=black] {};
        \node[fit=(n18), draw, fill=black] {};


        \pedge{1}{0};
        \pedge{2}{0};
        \pedge{3}{2};
        \pedge{4}{0};
        \pedge{5}{4};
        \pedge{6}{4};

        \pedge{10}{6};
        \pedge{14}{13};
        \pedge{18}{14};
        \pedge{22}{18};

        \pedge{38}{34};
        \pedge{42}{38};
        \pedge{46}{45};
        \pedge{50}{46};

        \pedge{54}{53};
        \pedge{55}{54};
        \pedge{56}{53};
        \pedge{57}{50};
        \pedge{58}{57};
        \pedge{59}{4};
        \pedge{60}{59};
        \pedge{61}{60};
        \pedge{62}{60};
        \pedge{63}{0};
        \pedge{64}{63};

        \foreach[count=\i from 1] \startnode in {n6.south west,n50.south west,n54.south west,n64.south east} {
            \draw[white, line width=3pt] (\startnode) to ++(0, -5);
            \draw[densely dashed] (\startnode) to
            node[pos=1, inner sep=3pt] (dl\i) {} ++(0, -4.5);
        }

        \draw[thick, decorate, decoration={brace, amplitude=5pt, mirror}] (dl1.north east) -- node[midway, below=.6, unitnode] (lbr) {} (dl2.north west);
        \draw[thick, decorate, decoration={brace, amplitude=5pt, mirror}] (dl3.north east) -- node[midway, below=.6, unitnode] (rbr) {}(dl4.north west);

        \draw[Latex-, red, thick] (lbr.north) to (lbr.south) to[out=270, in=270, looseness=.2] 
        node[midway, fill=white] {\ding{55}}
        node[midway, below=.75, black] {\small$\forall i \in [c_r + p, \absolute{Y}] :  \psv_Y(i) \notin [c_1, c_r)$} 
        (rbr.south) to (rbr.north);

        \node[fit=(nc5)(nc-4)] (ctr) {};
        \node at (ctr) {$\cdots$};
    \end{tikzpicture}
    \caption{An illustration of the settings in \cref{lem:no_sub_in_run,lem:trim} with $k=3$.
    The $\psv_Y$ array is shown with grey arrows below string $Y$.
    \cref{lem:no_sub_in_run} establishes that any sequence of $\hdist(X \leadsto Y)$ substitutions that transforms $X$ into $X'$ such that ${X' \approx Y}$ does not modify the fragment highlighted in red; the condition ``$\forall i \in [c_r + p, \absolute{Y}] :  \psv_Y(i) \notin [c_1, c_r)$'' ensures that none of these substitutions ``interact''  with this fragment.
    \cref{lem:trim} establishes that we can trim a pair of fragments (one in $X$ and one in $Y$) without altering the value of $\hdist(X \leadsto Y)$.}
    \label{fig:placeholder}
\end{figure}

\begin{lemma}\label{lem:no_sub_in_run}
    Consider $p, r, k \in \mathbb Z_{+}$ and strings $X = X_1UX_2$ and $Y = Y_1VY_2$ such that all of the following hold:
    \smallskip
    \begin{itemize}
        \item $\absolute{X_1} = \absolute{Y_1}$, $\absolute{X_2} = \absolute{Y_2}$, and $\absolute{U} = \absolute{V} = rp$, where $r \geq 2k + 1$;
        \item $U \approx V$;
        \item there is no position $i \in [\absolute{Y_1V} + 1, \absolute{Y}]$ such that $\psv_Y(i) \in [\absolute{Y_1} + 1, \absolute{Y_1V}-p]$;
        \item if we define $\forall i \in [1, r] : c_i := \absolute{X_1} + (i-1) p + 1$, then $X[c_i]$ is a minimum in $X[c_i\dd c_r + p)$, and $Y[c_i]$ is a minimum in $Y[c_i\dd c_r + p)$;
        \item $\hdist(X \leadsto Y) \leq k$.
    \end{itemize}
    \smallskip
    Any sequence of $\hdist(X \leadsto Y)$ substitutions that transforms $X$ into $X'$ with $X' \approx Y$ does not perform any substitution in fragment $X[c_{k + 1} \dd c_{r - k}]$.
\end{lemma}

\begin{proof}
    Consider any sequence of $\hdist(X \leadsto Y)$ substitutions that transforms $X$ into $X'$ such that ${X' \approx Y}$.
    We fix $i,j \in [1,r]$ to be the minimal and maximal values, respectively, such that $X[c_{i}]$ and $X[c_{j}]$ are not modified. (There are at least two distinct values because $r \geq \hdist(X \leadsto Y) + 2$.)
    We show that no substitution is performed in $X[c_i\dd c_j]$.

    Let $Z = X'[1 \dd c_i) \cdot X[c_i \dd c_j] \cdot X'(c_j \dd \absolute{X}] = X'[1 \dd c_i] \cdot X(c_i \dd c_j) \cdot X'[c_j \dd \absolute{X}]$. We claim that $Z \approx Y$.
    If not, there would exist $x \in [1, \absolute{X}]$ with $\psv_Z(x) \neq \psv_Y(x)$. We show that such $x$ does not exist by distinguishing between the following four cases:
    \begin{description}
        \item[Case 1:] $x \in [1, c_i]$. Since $Z$ has prefix $X'[1 \dd c_i]$, and due to $X' \approx Y$, \cref{fact:subeq}, and \cref{lem:psv_equiv}, we have $\forall x \in [1, c_i] : \psv_Z(x) = \psv_{X'}(x) = \psv_Y(x)$.
        \item[Case 2:] $x \in (c_i, c_j]$. $Y[c_i]$ is minimal in $Y[c_i \dd c_r]$ by assumptions.
        Hence, $\psv_Y(x) \geq c_i$. We obtain $\psv_Z(x) = \psv_Y(x)$ due to \cref{lem:psv_equiv} and the observation that:
        \[
        Z[c_i \dd c_j] = X[c_i \dd c_j] = {} U(ip - p \dd jp-p + 1]
        {}\approx {} V(ip - p \dd jp-p + 1] = Y[c_i \dd c_j].
        \]
        \item[Case 3:] $x \in (c_j,c_r + p)$. Since $Y[c_j]$ is minimal in $Y[c_j\dd c_r + p)$, we have $\psv_Y(x) \geq c_j$. Therefore, $\psv_Z(x) = \psv_{X'}(x) = \psv_Y(x)$ follows from
        $Z[c_j\dd c_r + p) = X'[c_j\dd c_r + p) \approx Y[c_j\dd c_r + p)$.
        \item[Case 4:] $x \in [c_r + p, \absolute{X}]$.
        By the assumptions, we have $\psv_Y(x) \notin [\absolute{Y_1} + 1, \absolute{Y_1V}-p] = [c_1, c_r)$. 
        First, consider the case when $\psv_Y(x) \geq c_r$.
        Note that $c_j \leq c_r$ implies $Z[c_r\dd \absolute{X}] = X'[c_r\dd \absolute{X}]$. Hence, $\psv_Z(x) = \psv_{X'}(x) = \psv_Y(x)$ follows from $X'\approx Y$, together with \cref{fact:subeq,lem:psv_equiv}.
        If, however, $\psv_{X'}(x) = \psv_Y(x) < c_1 \leq c_i$, then clearly $X'[x] < X'[c_i] = X[c_i]$.
        Since $X[c_i] = Z[c_i]$ is minimal in $X[c_i\dd c_j) = Z[c_i\dd c_j)$, it cannot be that $\psv_Z(x) \in [c_i, c_j)$, which implies $\psv_Z(x) = \psv_{X'}(x) = \psv_Y(x)$.
    \end{description}
    We have shown that $Z \approx Y$.
    Hence, the considered sequence of substitutions did not modify $X[c_i\dd c_j]$.
    Finally, by the definition of $i$ and $j$, the characters in positions $c_1, \dots, c_{i - 1}$ and $c_{j + 1}, \dots, c_r$ in $X$ must have been modified.
    We therefore have $(i - 1) + (r - j) \leq \hdist(X \leadsto Y) \leq k$, which implies $i \leq k+1$ and $j \geq r - k$.
    Consequently, there are no substitutions in $X[c_{k + 1} \dd c_{r - k}]$.
\end{proof}

\begingroup
\def\ell{a}
\def\err{b}
\begin{lemma}\label{lem:delete_retain_ct_eq}
    Consider strings $X$ and $Y$ with $X \approx Y$.
    Let $\ell, \err \in [1, \absolute{Y}]$ such that $\ell \leq \err$ and $Y[\ell]$ is minimal in $Y[\ell\dd \err]$.
    Then $X[1\dd \ell] \cdot X(\err\dd \absolute{X}] \approx Y[1\dd \ell] \cdot Y(\err\dd \absolute{X}]$.
\end{lemma}

\begin{proof}
    The assumption that $X \approx Y$, together with \cref{fact:subeq,lem:psv_equiv}, imply that $X[\ell]$ is minimal in $X[\ell\dd \err]$.
    For $Z \in \{X, Y\}$, let $\hat Z = Z[1\dd \ell] \cdot Z(\err\dd \absolute{Z}]$. It is easy to see that the following hold:
    \begin{itemize}
        \item For $i \in [1, \ell]$,  $\psv_{\hat Z}(i) = \psv_{Z}(i)$. (This is trivial.)
        \item For $i \in (\ell, \absolute{\hat Z}]$ with $\psv_Z(i + \err - \ell) \in (\err, \absolute{Z}]$,  we have $\psv_{\hat Z}(i) = \psv_Z(i + \err - \ell) - \err + \ell$. 
        \item For $i \in (\ell, \absolute{\hat Z}]$ with $\psv_Z(i + \err - \ell) \in (\ell, \err]$, we have $\psv_{\hat Z}(i) = \ell$. 
        \item For $i \in (\ell, \absolute{\hat Z}]$ with $\psv_Z(i + \err - \ell) \in [1, \ell]$, we have  $\psv_{\hat Z}(i) = \psv_Z(i + \err - \ell)$. 
    \end{itemize}
    Note that the definition of $\psv_{\hat Z}$ depends only on the $\psv_Z$ array. Since $X \approx Y$, we have that $\psv_X = \psv_Y$ by \cref{lem:psv_equiv}.
    Therefore, it follows that $\psv_{\hat X} = \psv_{\hat Y}$.
    By applying \cref{lem:psv_equiv} again, we conclude that $\hat X \approx \hat Y$.
\end{proof}
\endgroup

\begin{lemma}\label{lem:trim}
    Consider $p, r, k \in \mathbb Z_{+}$ and strings $X = X_1UX_2$ and $Y = Y_1VY_2$ such that all of the following hold:
    \smallskip
    \begin{itemize}
        \item $\absolute{X_1} = \absolute{Y_1}$, $\absolute{X_2} = \absolute{Y_2}$, and $\absolute{U} = \absolute{V} = rp$, where $r \geq 2k + 3$;
        \item $U \approx V$;
        \item there is no position $i \in [\absolute{Y_1V} + 1, \absolute{Y}]$ such that $\psv_Y(i) \in [\absolute{Y_1} + 1, \absolute{Y_1V}-p]$;
        \item if we define $\forall i \in [1, r] : c_i := \absolute{X_1} + (i-1)p + 1$, then $X[c_i]$ is a minimum in $X[c_i\dd c_r + p)$, and $Y[c_i]$ is a minimum in $Y[c_i\dd c_r + p)$.
    \end{itemize}
    \smallskip
    Then $\hdist_k(\hat X \leadsto \hat Y) = \hdist_k(X \leadsto Y)$, where
    \[\hat X = X[1\dd c_{k + 1}] \cdot X(c_{r - k} - p\dd \absolute{X}]\quad\text{and}\quad\hat Y = Y[1\dd c_{k + 1}] \cdot Y(c_{r - k} - p\dd \absolute{Y}].\]
\end{lemma}

\begin{proof}
    It suffices to prove the following two claims, which together imply the statement.

    \begin{claim}\label{claim:st1}
        $\hdist(X \leadsto Y) \leq k$ implies $\hdist(\hat X \leadsto \hat Y) \leq \hdist(X \leadsto Y)$.
    \end{claim}
    \begin{claim}\label{claim:st2}
        $\hdist(\hat X \leadsto \hat Y) \leq k$ implies $\hdist(\hat X \leadsto \hat Y) \geq \hdist(X \leadsto Y)$.
    \end{claim}

    \begin{claimproof}[Proof of \cref{claim:st1}]
    Assume that $\hdist(X \leadsto Y) \leq k$ and consider any sequence of $\hdist(X \leadsto Y)$ substitutions that transform $X$ into a string $X'$ with $X' \approx Y$.
    By \cref{lem:no_sub_in_run}, we have $X' = X'[1\dd c_{k + 1}) \cdot X[c_{k + 1}\dd c_{r - k}] \cdot X'(c_{r - k} \dd \absolute{X}]$. 
    Note that we can perform the same sequence of substitutions (with an appropriate shift) on $\hat X$ as well.
    Hence, there is a sequence of $\hdist(X \leadsto Y)$ substitutions that transforms $\hat X$ into 
    \[\hat X' : = X'[1\dd c_{k + 1}) \cdot X[c_{k + 1}] \cdot X(c_{r - k} - p \dd c_{r - k}] \cdot X'(c_{r - k} \dd \absolute{X}].\]%
    It remains to be shown that $\hat X' \approx \hat Y$, which implies $\hdist(\hat X \leadsto \hat Y) \leq \hdist(X \leadsto Y)$.
    Note that $\hat X'$ and $\hat Y$ can be obtained by deleting fragments $X'(c_{k + 1}\dd c_{r - k} - p]$ and $Y(c_{k + 1}\dd c_{r - k} - p]$ from $X'$ and $Y$, respectively.
    By the definition of $Y$, we know that $Y[c_{k + 1}]$ is a minimum in $Y[c_{k + 1}\dd c_{r - k} - p]$. Hence, we can apply \cref{lem:delete_retain_ct_eq} to $X'$ and $Y$ with $a = c_{k + 1}$ and $b = c_{r - k} - p$, thus obtaining $\hat X'\approx \hat Y$.
    \end{claimproof}
    
    \begin{claimproof}[Proof of \cref{claim:st2}]
    Assume that $\hdist(\hat X \leadsto \hat Y) \leq k$ and consider any sequence of $\hdist(\hat X \leadsto \hat Y)$ substitutions that transform $\hat X$ into a string $\hat X'$ with $\hat X'\approx \hat Y$.
    It can be readily verified that the strings $\hat X$ and $\hat Y$ still satisfy the conditions of \cref{lem:no_sub_in_run} with parameter $\hat r = 2k + 2$.
    Therefore, we have: 
    \begin{alignat*}{1}
        \hat X' = {}& \hat X'[1\dd c_{k + 1}) \cdot \hat X[c_{k + 1}\dd c_{\hat r - k}] \cdot \hat X'(c_{\hat r - k} \dd \absolute{\hat X}]\\
        {} = {}&\hat X'[1\dd c_{k + 1}) \cdot X[c_{k + 1}] \cdot X(c_{r - k} - p \dd c_{r - k}] \cdot \hat X'(c_{k + 2} \dd \absolute{\hat X}].
    \end{alignat*}
    We observe that the same sequence of substitutions can also be applied directly to $X$, resulting in the string:
    $X' = \hat X'[1\dd c_{k + 1}) \cdot X[c_{k + 1}\dd c_{r - k}] \cdot \hat X'(c_{k + 2} \dd \absolute{\hat X}]$.
    
    It remains to be shown that $X' \approx Y$, which implies $\hdist(\hat X \leadsto \hat Y) \geq {\hdist(X \leadsto Y)}$. Note that $\hat X'$ and $X'$, as well as $\hat Y$ and $Y$, share a prefix of length $c_{k + 1}$.
    Due to ${\hat X' \approx \hat Y}$, combined with \cref{fact:subeq,lem:psv_equiv}, it is clear that, for every $x \in [1, c_{k + 1}]$, $\psv_{X'}(x) = \psv_{\hat X'}(x) = \psv_{\hat Y}(x) = \psv_Y(x)$.
    
    Next, we observe that $X'[c_{k + 1}\dd c_{r - k}] = X[c_{k + 1}\dd c_{r - k}]$. Recall that $X[c_{k + 1}]$ and $Y[c_{k + 1}]$ are minimal in $X[c_{k + 1}\dd c_{r - k}]$ and $Y[c_{k + 1}\dd c_{r - k}]$, respectively.
    Due to $U \approx V$, combined with \cref{fact:subeq,lem:psv_equiv}, we have $\forall x \in (c_{k + 1}, c_{r - k}] : \psv_{X'}(x) = \psv_Y(x) \geq c_{k + 1}$.
    
    Finally, we consider $x \in (c_{r - k}, \absolute{X}]$.
    Let $\hat x = x - \absolute{X} + \absolute{\hat X}$.
    We observe that $\hat X'$ and $X'$, as well as $\hat Y$ and~$Y$, share a suffix of length $\ell := \absolute{X} - c_{r - k} + p$.
    If $\psv_Y(x) > c_{r - k} - p$, then $\hat X' \approx \hat Y$ already implies $\psv_{\hat X'}(\hat x) = \psv_{\hat Y}(\hat x)$.
    It remains to consider the case when $x \in (c_{r - k}, \absolute{X}]$ and $\psv_Y(x) \leq c_{r - k} - p$. In this case, since $Y[c_{r - k}]$ is minimal in $Y[c_{r - k} \dd c_r + p)$, it holds that $x \geq c_r + p = \absolute{Y_1V} + 1$, and, by the condition of the lemma, $\psv_Y(x) \leq \absolute{Y_1} = c_1 - 1$.
    Note that $\psv_{\hat Y}(\hat x) = \psv_Y(x)$, as we obtained $\hat Y$ by merely deleting some elements between $\psv_Y(x)$ and~$x$.
    Due to $\hat X' \approx \hat Y$, we have $\psv_{\hat X'}(\hat x) = \psv_{\hat Y}(\hat x)$.
    Finally, we obtain $X'$ from $\hat X'$ by inserting into $\hat X'$ values that are at least as large as $\hat X'[c_{k + 1}] = X[c_{k + 1}]$ between positions $\psv_{\hat X'}(\hat x)$ and $\hat x$.
    Since $\hat X'[c_{k + 1}] \in (\psv_{\hat X'}(\hat x), \hat x)$, we know that $\hat X'[\hat x] < \hat X'[c_{k + 1}]$, and thus all inserted values are at least $\hat X'[\hat x]$.
    Thus, $\psv_{X'}(x) = \psv_{\hat X'}(\hat x) = \psv_{\hat Y}(\hat x) = \psv_Y(x)$, as required.
    \end{claimproof}
    
    The combination of the two claims yields that either both $\hdist(\hat X \leadsto \hat Y)$ and $\hdist(X \leadsto Y)$ are at most $k$ and are in fact equal, or they are both greater than $k$.
\end{proof}

A symmetric claim holds for \emph{next-smaller-values}. For completeness, we provide a proof of the corresponding lemma below in 
\cref{app:symmetry}.

\begin{restatable}[analogue of \cref{lem:trim}]{lemma}{restatetrimnsv}\label{lem:trim:nsv}
    Consider $p, r, k \in \mathbb Z_{+}$ and strings $X = X_1UX_2$ and $Y = Y_1VY_2$ such that all of the following hold:
    \smallskip
    \begin{itemize}
        \item $\absolute{X_1} = \absolute{Y_1}$, $\absolute{X_2} = \absolute{Y_2}$, and $\absolute{U} = \absolute{V} = rp$, where $r \geq 2k + 3$;
        \item $U \approx V$;
        \item there is no position $i \in [1, \absolute{Y_1}]$ such that $\nsv_Y(i) \in [\absolute{Y_1} + p + 1, \absolute{Y_1V}]$;
        \item if we define $\forall i \in [1, r] : c_i := \absolute{X_1} + ip$, then $X[c_i]$ is a minimum in $X(c_1 - p\dd c_i]$, and $Y[c_i]$ is a minimum in $Y(c_1 - p\dd c_i]$.
    \end{itemize}
    \smallskip
    Then $\hdist_k(\hat X \leadsto \hat Y) = \hdist_k(X \leadsto Y)$, where
    \[\hat X = X[1\dd c_{k + 1} + p) \cdot X[c_{r - k}\dd \absolute{X}]\quad\text{and}\quad\hat Y = Y[1\dd c_{k + 1} + p) \cdot Y[c_{r - k}\dd \absolute{Y}].\]
\end{restatable}

We are now ready to prove the main algorithmic result of this section.

\begin{lemma}\label{lem:per_case}
    Let $n,m,k,p,\Delta \in \mathbb{Z}_+$ satisfy $n\leq 3m/2$, $p<\Delta$, $k<\Delta$, and $k\Delta\leq\sqrt{m}/128$.
    Given $k$, $p$, $\Delta$, a text $T[1 \dd n]$, a pattern $P[1 \dd m]$, and a fragment $P[\ell \dd r)$ of length at least $m/{(4\Delta)}$ with minimum CT-block-period $p$, we can compute ${\hdist_k(T[i \dd i+m) \leadsto P)}$ for all $i\in[1,n-m+1]$ in $\cO(mk^4\Delta)$ time. 
\end{lemma}

\begin{proof}
    Due to $\Delta^2\leq m/(128^2k^2)$, the fragment $P[\ell \dd r)$ is of length at least $m/{(4\Delta)}\geq {p \cdot m/(4\Delta^2)} > p\cdot 1000k^2$ and hence consists of at least $1000k^2$ blocks of length $p$; namely, $P[\ell + ip \dd \ell +(i+1)p)$ for $i \in [0, (r-\ell)/p)$.
    In any $k$-approximate CT-occurrence of $P[\ell \dd r)$ in $T$, at least $1000 k^2-k \geq 999 k^2$ of those blocks are aligned with  fragments of $T$ that CT-match them.
    These blocks are partitioned into at most $k+1$ sets of contiguous blocks by the at most $k$ blocks that are aligned with  fragments of $T$ that do not CT-match them.
    Hence, in any $k$-approximate CT-occurrence of $P[\ell \dd r)$ in $T$, there exists a fragment~$F$ of $P[\ell \dd r)$ that consists of at least $999k^2/(k+1) > 100k$ blocks and is CT-matched exactly.\footnote{Note that it might be a different fragment for each $k$-approximate CT-occurrence of $P[\ell \dd r)$ in $T$.}
    We have $\absolute{F}=\Omega(\absolute{P[\ell \dd r)]}/k)\supseteq \Omega(m / (k\Delta))$.
    Further, the CT-block-periodicity of $P[\ell \dd r)$ implies that
    $F \approx P[\ell \dd \ell+\absolute{F})$ and, combined with \cref{lem:per_frag}, that~$F$ has minimum CT-block-period~$p$.

   For a string $X$, a block $X[i \dd i+p)$ of length $p$ with $i-p, i+2p-1 \in [1,\absolute{X}]$, is \emph{bad} if $X[i-p \dd i+2p) \not\approx P[\ell \dd \ell+3p)$ and \emph{good} otherwise.
   Intuitively, if we have many bad blocks in $P$, at least one of them has to be aligned with a bad block in the text.
   In this case, we exploit this property to obtain a few candidate starting positions of $k$-approximate CT-matches. We then verify each of them using \cref{lem:compare}.
   In the complementary case, the pattern is almost periodic (that is, it almost entirely consists of good blocks).
   We exploit this by trimming long pairs of fragments (one in $T$ and one in $P$) that are composed solely of good blocks using either \cref{lem:trim} or \cref{lem:trim:nsv} before invoking \cref{lem:compare}.
    
    We compute a fragment $\Phi = P[\ell' \dd r')$ of $P$ by extending $P[\ell \dd r)$ blockwise to the left (and to the right) by prepending (resp.~appending) blocks of $p$ characters until we either accumulate $3k + 1$ bad blocks in that direction or there are fewer than $2p$ characters left in that direction---in the latter case, we almost reach the beginning (resp.~end) of $P$.

    \input{figs/blocks}
    
    \subparagraph{Case I: Enough blocks accumulated in at least one direction.}
    Assume, without loss of generality, that we accumulated $3k+1$ bad blocks while extending $P[\ell \dd r)$ to the left.
    
    We employ the linear-time exact CT-matching algorithm of \cref{lem:exact} to compute all exact CT-occurrences of $F := P[\ell \dd \ell+100kp)$ in $T$.
    By \cref{lem:not_folklore}, these CT-occurrences can be partitioned into $\cO(n / \absolute{F}) \subseteq \cO(k\Delta)$ arithmetic progressions with common difference $p$ such that the CT-occurrences in each of the progressions span a CT-run of length $\Omega(m / (k\Delta))$ with minimum CT-block-period $p$. Let $\mathcal{R}$ denote the set of CT-runs in $T$ that are spanned by the exact CT-occurrences of $F$.
    
    Let us fix a CT-run $T[a \dd b] = R \in \mathcal{R}$. We extend $R$ to the left using  a procedure identical to the one  performed on $P[\ell \dd r)$ until we either accumulate $3k+1$ bad blocks or the fragment contains $T[2p]$.
    Consider aligning the pattern with a fragment $T[i \dd i+m)$ such that a fragment of $P[\ell\dd r)$ that CT-matches $F$ is aligned with a fragment of $R$ that CT-matches it.
    By \cref{lem:not_folklore}, any such alignment satisfies $a - i + 1 \equiv \ell \pmod p$, that is, it perfectly aligns the block boundaries of $P[\ell \dd r)$ and $R$.
    In such an alignment, if a bad (resp.\ good) block of the pattern is aligned with a good (resp.~bad) block of the text, at least one substitution is required in either that text block or one of its two neighbours in any sequence of substitutions that transforms $T[i \dd i+m)$ into a string $X \approx P$.
    Thus, $k$ substitutions performed on $T[i \dd i+m)$ can fix at most $3k$ misaligned pairs (that is, pairs consisting of a bad block in one of the strings and a good block in the other).
    Therefore, at least one of the leftmost $9k + 1$ bad blocks of $\Phi$ must be aligned with one of the at most $3k + 1$ bad blocks in the extension of~$R$.
    There are only $\cO(k^2)$ ways to align such a pair of bad blocks, each yielding a candidate starting position of a $k$-approximate CT-occurrence of $P$ in $T$.
    The described procedure can be implemented naively in $\cO(m)$ time.
    
    We perform this procedure for all CT-runs in $\mathcal{R}$, thus obtaining $\cO(\absolute{\mathcal{R}} \cdot k^2) = \cO(k^3\Delta)$ candidate starting positions.
    We then invoke \cref{lem:compare}
    for each such position.
    In total, this takes time $\cO(k^3\Delta \cdot mk) = \cO(mk^4\Delta)$.

    \subparagraph{Case II: The pattern is almost CT-block-periodic.} In the complementary case, we did not accumulate $9k + 1$ bad blocks in either direction. Thus, $P = X_P B_P Y_P$ where $\absolute{X_P},\absolute{Y_P} < 2p$ and~$B_P$ consists of $p$-length blocks, at most $6k$ of which are bad.
    
\begin{claim}\label{claim:trimt}
    In $\cO(n)$ time, we can either decide that $P$ has no $k$-approximate CT occurrences in $T$ or compute
    a fragment $T' = T[t_1 \dd t_2) = X_T B_T Y_T$ such that the following hold:
        \begin{itemize}
            \item for all $i \in [1, n-m+1]$, $\hdist_k(T[i \dd i+m) \leadsto P)\leq k$, implies that $i,i+m \in [t_1, t_2]$ and $i \equiv t_1 \pmod p$.
            \item $\absolute{X_T} = \absolute{X_P}$ and $\absolute{Y_T} = \absolute{Y_P}$,
        \end{itemize}
    Moreover, $B_T$ is partitioned into $p$-length blocks $\cO(k)$ of which are explicitly marked as bad.
\end{claim}
\begin{claimproof}
    By \cref{lem:per_frag}, the fragment $\Pi:=P[\ell \dd \ell+3p)$ has minimum CT-block-period $p$.
    Thus, by  \cref{lem:not_folklore}, any two exact CT-occurrences of $\Pi$ in $T$ are at distance at least $p$ and there are at most $n/p$ CT-occurrences in total.
    
    Now, any $k$-approximate CT-occurrence of $P$ in $T$ implies a large number of exact CT-occurrences of $\Pi$ at positions that are in a single equivalence class modulo $p$, with each such CT-occurrence corresponding to aligning a good block of $B_P$ with a fragment of $T$.
    Specifically, by accounting for the at most $k$ substitutions (each potentially affecting three good blocks), the at most $6k$ bad blocks in $P$, as well as $X_P$ and $X_T$ which are of total length at most $4p$, we conclude that we must have at least
    ${\floor{(m-4p)/p}} - 6k - 3k \geq m/p - 13k$ such exact CT-occurrences.
    The inequalities $k\Delta \leq \sqrt{m}/128$ and $p < \Delta$ imply that $13kp < \sqrt{m}/4$, and hence $m/p - 13k \geq (m-\sqrt{m}/4)/p$.
    Since we have a total of $n/p$ CT-occurrences of $\Pi$ in $T$ and $n/p \leq 3m/(2p) < 2(m-\sqrt{m}/4)/p$, there is at most one residue modulo $p$, say~$\rho$, with at least $m/p - \sqrt{m}/p$ CT-occurrences of $\Pi$ starting at a positions equivalent to $\rho$ modulo~$p$.
    
    We compute all exact CT-occurrences of $\Pi$ in $T$ in $\cO(n)$ time using \cref{lem:exact} and identify the majority residue modulo $p$, say $\rho$, resolving ties arbitrarily.
    If the CT-occurrences of $\Pi$ at positions of $T$ equivalent to $\rho$ modulo $p$ are fewer than $m/p - \sqrt{m}/p$, we conclude that~$P$ does not have any $k$-approximate CT-occurrences in $T$ and terminate the algorithm. We henceforth, assume that we are in the complementary case.    
    
    We select $y$ as the leftmost position in $T$ such that $y \geq n/2$ and $y \equiv \rho \pmod p$. We extend the fragment $T[y\dd y+p)$ in each direction until we either accumulate $9k + 1$ bad blocks or there are fewer than $2p$ characters left. Let $T[a \dd b)$ denote the obtained fragment.

    Any $k$-approximate CT-occurrence $T[i \dd i+m)$ of $P$ in $T$ must align $B_P$ with a fragment of the form $T[a + ip \dd b -jp)$ for nonnegative integers $i$ and $j$.
    Further, in this CT-occurrence,~$B_P$ must be aligned with a fragment of $T$ that contains $T[\floor{n/2}]$.
    Now, observe that at most~$9k$ bad blocks of $T$ can be contained in $T[i \dd i+m)$, as at most $6k$ of them can be aligned with CT-matching bad blocks of $P$ and the at most $k$ substitutions may fix at most $3k$ pairs of bad blocks in $T$ with good blocks in $P$.
    This implies that $i+\absolute{X_P}, i+\absolute{X_P B_P} \in [a,b]$.
        
    Next, we define the boundaries $t_1$ and $t_2$, and describe the decomposition $T' = X_T B_T Y_T$.
    Let $t_1$ be the smallest positive integer in $\{ a + ip - \absolute{X_P}: i \in \{0,1,2\} \}$, let $t_2$ be the largest integer in $\{ b - ip + \absolute{Y_P}: i \in \{0,1,2\} \} \cap [1,n]$, and set $T'=T[t_1 \dd t_2)$.
    Since, for any approximate CT-occurrence $T[i \dd i+m)$ we have $i+\absolute{X_P}, i+\absolute{X_P B_P} \in [a,b]$, this ensures that $i,i+m \in [t_1,t_2]$. Further 
    we have that $i \equiv t_1 \pmod p$ since $i+\absolute{X_P} \equiv a \pmod p$.
    The partition of $T'$ is then defined as:
    $X_T = T[t_1 \dd t_1 + \absolute{X_P})$,
    $B_T = T[t_1 + \absolute{X_P} \dd t_2 - \absolute{Y_P})$, $Y_T = T[t_2 - \absolute{Y_P} \dd t_2)$.
    Finally, we return $t_1$, $t_2$, and the bad blocks in $B_T$.
\end{claimproof}
    
    We also designate each block in $B_P$ as \emph{implicitly bad} if it contains a position $i$ such that there exists a~position $j$ in either an explicitly bad block or in $J:=[1 , \absolute{X_P}] \cup (\absolute{X_PB_P},\absolute{P}]$, such that $i=\psv_P(j)$ or $i=\nsv_P(j)$.
    Note that the number of implicitly bad blocks is $\cO(kp) \subseteq \cO(k\Delta)$ (since $p<\Delta$) as each (explicitly) bad block implies $\cO(p)$ implicitly bad ones and $J$ implies $\cO(p)$ implicitly bad blocks.
    This designation is necessary because substitutions are local operations, whereas CT-equivalence is a global property. Although substitutions directly modify values within explicitly bad blocks, a single substitution can indirectly affect many distant positions. In particular, if a substitution occurs at position~$i$  such that  $i= \psv_P(j)$ or $i= \nsv_P(j)$ for some position $j$, then the values $\psv_P(j)$ or $\psv_P(j)$ may change.
    To ensure that \cref{lem:trim} can be invoked to trim long sequences of good blocks, we must account for all blocks that could potentially interact with either the boundary fragments or the blocks aligned with blocks that contain substitutions.

    By \cref{claim:trimt}, it suffices to compute $\hdist_k(T[i \dd i+m) \leadsto P)$ for each of the $\cO(m / p)$ positions in $T$ that satisfy $i \equiv t_1 \pmod p$ and $i,i+m \in [t_1, t_2]$.
    Let $G =T[i \dd i+m)$.
    $G$~inherits its block-decomposition from $B_T$, as well as the classification of some blocks as bad.
    We call the considered position $i$ an \emph{overlap position} if there exist positions $\pi$ in $P$ and $\phi$ in $G$ that belong to  (implicitly or explicitly) bad blocks and satisfy $\absolute{\pi - \phi} \leq 6kp$, and a~\emph{non-overlap position} otherwise.
    We treat overlap and non-overlap positions differently:
    \begin{itemize}
        \item If $i$ is an overlap position, we employ \cref{lem:compare} to compute $\hdist_k(G \leadsto P)$ in $\cO(mk)$ time.
        The number of overlap positions is $\cO(k^3p) = \cO(k^{3}\Delta)$ since we require one of the $\cO(kp)$ (implicitly or explicitly) bad blocks in $B_P$ to be within $\cO(k)$ blocks of one of the $\cO(k)$ bad blocks in $B_T$.
        The total time required for overlap positions is thus $\cO(mk^{4}\Delta)$.
        Note that we can efficiently enumerate all the overlap positions using the list of bad blocks for both $B_P$ and $B_T$.
        
        \item If $i$ is a non-overlap position, then we can show that $\hdist_k(G \leadsto P)$ only depends on the parts of the alignment of $G$ and $P$ in the $\cO(k)$-block vicinity of the (implicitly or explicitly) bad blocks. More precisely, \cref{lem:trim,lem:trim:nsv} will allow us to delete most of the good blocks, significantly reducing the lengths of $G$ and $P$ prior to performing the comparison of \cref{lem:compare}. We first describe this trimming procedure, and then justify that the conditions for \cref{lem:trim,lem:trim:nsv} are indeed satisfied.
        
        We use \cref{lem:trim} if $P[\ell']$ is the leftmost minimum of $\Phi$ and \cref{lem:trim:nsv} otherwise, that is, if the leftmost minimum is $P[r'-1]$.
        Conceptually, we perform the following procedure.
        We traverse the lists of bad blocks in each of $G$ and~$P$ from left to right in parallel in $\cO(k)$ time.
        Whenever we encounter a maximal pair of fragments $U = G[x_1 \dd x_2)$ and $V=P[x_1 \dd x_2)$ of length at least $(6k+1)p$ that only consist of blocks that are not (implicitly or explicitly) bad, we replace $G[x_1 \dd x_2)$ with $U[1 \dd 3kp] \cdot U(\absolute{U} - 3kp \dd \absolute{U}]$ and $P[x_1 \dd x_2)$ with $V[1 \dd 3kp] \cdot V(\absolute{V} - 3kp \dd \absolute{V}]$.
        A repeated application of either \cref{lem:trim} or \cref{lem:trim:nsv} guarantees that $\hdist_k(G \leadsto P) = \hdist_k(G' \leadsto P')$, where $G'$ and $P'$ are the strings obtained from $G$ and $P$, respectively.
        Instead of making the replacements on $G$ and $P$, we explicitly construct $G'$ and $P'$ by copying the fragments of $G$ and $P$ that are not deleted by the conceptual procedure described above.

        As we have $\cO(kp)$ (implicitly or explicitly) bad blocks in $P$ and $\cO(k)$ bad blocks in $G$, and since each bad block prevents $\cO(k)$ blocks from being trimmed, we have $\absolute{G'}=\absolute{P'}=\cO(k^2p^2)$.
        The computation of $\hdist_k(G \leadsto P)$ using \cref{lem:compare} thus takes $\cO(k^2p^2 \cdot k)$ time and returns $\hdist_k(G \leadsto P)$.
        Over all $\cO(m/p)$ non-overlap positions,
        the total time required is 
        $\cO(m / p \cdot k^2p^2 \cdot k) = \cO(mpk^3) = \cO(mk^{3}\Delta)$.

        It remains to be shown that the described applications of \cref{lem:trim} are valid if $P[\ell']$ is the leftmost minimum of $\Phi$; applications of \cref{lem:trim:nsv} are valid in the remaining case by symmetric arguments.
        First, note that $p$ is a CT-block-period of $U$ and $V$ by \cref{lem:not_folklore}, and they both have their leftmost minimum at their first positions.
        \begin{itemize}
            \item  The first condition of \cref{lem:trim} holds trivially by our assumptions with $r = (x_2-x_1)/p$.
            \item The second condition of \cref{lem:trim} holds by \cref{cor:focusonp} as $p$ is a CT-block-period of both $U$ and $V$, and $U[1 \dd 2p] \approx P[\ell + p \dd \ell +3p) \approx V[1 \dd 2p]$.
            \item The third condition of \cref{lem:trim} holds due to the facts that (a) by the definition of implicitly bad blocks, there is no position $j$ outside a good block with $\psv_P(j) \in [x_1, x_2-p)$, and
            (b) for each position $j$ in a good block, $\psv_P(j)$ is at least $j-p$ by an application of \cref{lem:blockper_chain} to $\Phi$.
            \item The fourth condition of \cref{lem:trim} follows by applying \cref{lem:blockper_chain} to $U$ and $V$.\qedhere
        \end{itemize}
    \end{itemize}
\end{proof}

\subsection{Wrap-Up}

\mainthm*
\begin{proof}
We can assume that $k < m^{1/4}/128$; otherwise, the algorithm of Kim and Han~\cite{DBLP:journals/tcs/KimH25}, which simply invokes \cref{lem:compare} for $k$, $P$, and $T[i \dd i+m)$ for each $i\in [1,n -m+1]$, runs in $\cO(nmk) \subseteq \cO(nk^5)$ time.

For $\Delta \in \mathbb Z_+$ (to be specified later) satisfying $k\Delta \leq \sqrt{m}/128$, we apply \cref{lem:filtering_preprocess} to $P$. We either obtain a  CT-rainbow of  $\Delta$ disjoint fragments in $P$, or a fragment of length at least $m/{(4\Delta)}$ with its minimum CT-block-period $p < \Delta$.
This takes $\cO(m\Delta)$ time.

We then apply the so-called standard trick to reduce the given instance of \problemname to $\cO(n/m)$ instances of the same problem with the same pattern $P$, the same threshold $k$, and the text being a fragment of~$T$ of length at most~$3m/2$.
For $i \in [0, \floor{2n/m}]$, define
$T_i := T[1 + \floor{i \cdot m/2} \dd \min\{\floor{(i+3) \cdot m/2} , n+1\})$.
It is readily verified that each $m$-length fragment of $T$ appears in precisely one of the $T_i$s.

We then apply either \cref{lem:aper_case} or \cref{lem:per_case} (depending on the outcome of the analysis of the pattern according to \cref{lem:filtering_preprocess}) to solve an instance of \problemname separately for each $T_i$.
In the end, we merge the results.

Each instance is solved in $\cO(mk^4\Delta + m^2k/\Delta)$ time.
The total complexity over all $\cO(n/m)$ instances is therefore $\cO((n/m) (mk^4\Delta + m^2k/\Delta)) = \cO(nk^4\Delta + nmk/\Delta)$.
We now distinguish between two cases:
\begin{itemize}
    \item If $k \leq \floor{m^{1/5} / 128^{2/5}}$, we balance the two terms in the  complexity by setting $\Delta=\floor{\sqrt{m} / (128\cdot k^{1.5})}$.
    Since $k\Delta \leq \sqrt{m}/128$, and  $2k \leq \Delta$, the time complexity is
    \[\cO\left(nk^4 \cdot \frac{\sqrt{m}}{k^{1.5}} + \frac{nmk}{\sqrt{m}/k^{1.5}}\right) = \cO(n \sqrt{m} \cdot k^{2.5}).\]

    \item Else, we have $k \in (\floor{m^{1/5} / 128^{2/5}},\floor{m^{1/4}/128})$.
    We set $\Delta = 2k$, noting that $k\Delta = 2k^2$. Since $k < m^{1/4}/128$, we have $k^2 < \sqrt{m}/128^2$ and hence $k\Delta = 2k^2 < \sqrt{m}/128$. Hence, in this case, the time complexity is
    \[\cO\left(nk^4(2k) + \frac{nmk}{2k}\right) = \cO(nk^5 + nm) = \cO(nk^5).\]
    
\end{itemize}

Since $nk^5 = \Theta(n\sqrt{m} \cdot k^{2.5})$ for all $k \in [\floor{m^{1/5} / 128^{2/5}}, \floor{m^{1/5}}]$, the statement holds.
\end{proof}

\section{Symmetric Lemmas for Next Smaller Values}\label{app:symmetry}

The proof of the following lemma is symmetric to that of {\cite[Theorem 1]{exactCTmatching}}.

\begin{restatable}[analogue of \cref{lem:psv_equiv}]{lemma}{ndmapping}\label{lem:nsv_equiv}
    For equal-length strings $X$ and $Y$, $\CT(X) = \CT(Y)$ if and only if $\ND(X) = \ND(Y)$, or equivalently $\forall i \in [1, \absolute{X}] : \nsv_X(i) = \nsv_Y(i)$.
\end{restatable}

\begin{proof}
We prove the lemma  by induction on $n=\absolute{X}$.  

If $n = 1$, $X$ and $Y$ both have the same Cartesian tree consisting of a single node.  
Their $\ND$ representations are also identical ($\ND(X)[1] = \ND(Y)[1] = 0$ by definition).
Thus, the lemma holds for $n = 1$.

Assume the lemma holds for strings of length $i$. We show it holds for length $i+1$. 

($\Rightarrow$:) Assume $X[1 \dd i+1] \approx Y[1 \dd i+1]$. Two cases are to be considered:
\begin{itemize}
\item $X[i+1]$ and $Y[i+1]$ are not roots: Let $X[j]$ be the parent of $X[i+1]$, and $Y[\ell]$ the
parent of $Y[i+1]$. Since $X[1\dd i+1] \approx Y[1 \dd i+1]$, we must have $j = \ell$.  Removing the node at index $i+1$ results in $\CT(X[1 \dd i])$ and $\CT(Y[1 \dd i])$. By the induction hypothesis, $\ND(X[1 \dd i]) = \ND(Y[1 \dd i])$. 
Since $\nsv$ values represent the first smaller element to the right, and the structural position of $i+1$ relative to its parent $j$ is identical in both trees, it follows that $\nsv_X(z) = \nsv_Y(z)$ for all $z \in [1, i+1]$. Specifically, the ``new'' $\nsv$ values for indices in $[1, i]$ that were affected by the insertion of position $i+1$ are determined solely by the tree structure, which is identical. Hence, $\ND(X) = \ND(Y)$.

\item $X[i+1]$ and $Y[i+1]$ are roots: This implies $X[i+1]$ and $Y[i+1]$ are the  minima of their respective strings. Removing these roots leaves $\CT(X[1 \dd i])$ and $\CT(Y[1 \dd i])$ as the left subtrees. Since $\CT(X) = \CT(Y)$, these subtrees are identical, and by the induction hypothesis, $\ND(X[1 \dd i]) = \ND(Y[1 \dd i])$. Because the root is at $i+1$, any $z \leq i$ that had $\nsv_{X[1 \dd i]}(z) = i+1$ will maintain $\nsv_X(z) = i+1$. Thus, all $\ND$ values remain equal.
\end{itemize}

($\Leftarrow$:) Assume $\ND(X[1 \dd i+1]) = \ND(Y[1 \dd i+1])$. Since $\ND(X[1 \dd i]) = \ND(Y[1 \dd i])$ is implied by the consistency of the $\nsv$ definition, the induction hypothesis gives $X[1 \dd i] \approx Y[1 \dd i]$.

We can reconstruct $\CT(X[1 \dd i+1])$ from $\CT(X[1 \dd i])$ by considering $\nsv_X$. Let $z$ be the largest index such that $z < i+1$ and $\nsv_X(z) = i+1$.
\begin{itemize}
\item If such a $z$ exists, $X[i+1]$ is inserted as the right child of $X[z]$, and the previous right subtree of $X[z]$ becomes the left subtree of $X[i+1]$.
\item If no such $z$ exists, $X[i+1]$ becomes the new root, and the entire $\CT(X[1 \dd i])$ becomes its left subtree.
\end{itemize}
Since $X[1 \dd i] \approx Y[1 \dd i]$ and the $\nsv$ values for the new position $i+1$ are identical for both strings, the insertion process results in identical trees. Thus, $X[1 \dd i+1] \approx Y[1 \dd i+1]$.
\end{proof}

\begin{lemma}[analogue of \cref{lem:blockper_chain}]\label{lem:blockper_chain_v2}
    Consider a string $P[1 \dd m]$ that has a CT-block-period $p \in [1, \lfloor m/2\rfloor)$, where $P[m]$ is the  leftmost minimum of~$P$. 
    For $i \in [1, m/p]$, let $c_i = m - (i-1)p$.
    For all $i \in [1, m/p]$, it holds that $P[c_i]$ is the leftmost minimum  of $P[1 \dd c_i]$.
\end{lemma}

\begin{proof}
    For $i \in [1,m/p]$ denote $P_i := P[1 \dd c_i]$.
    Our goal is to show the more general statement that $p$ is a CT-block-period of $P_{i}$ for all $i$ with $P_{i}[\absolute{P_i}]$ being the leftmost minimum of~$P_{i}$.
    By our assumptions, this is true for $i=1$.
    Now, consider any $i \in [2,m/p]$ and suppose that $p$ is a CT-block-period of $P_{i-1}$ with $P_{i-1}[\absolute{P_{i-1}}]$ being the leftmost minimum of~$P_{i-1}$.
    First, since~$p$ is a CT-block-period of $P_{i-1}$, we have $P[1+p \dd c_{i-1}]= P_{i-1}[c_{i-1}-\absolute{P_i}+1 \dd c_{i-1}] \approx P_i$ and hence the  leftmost minimum of $P_i$ is $P_i[\absolute{P_i}]$.
    Further, \cref{fact:subeq} implies that $P[1 \dd c_i - p] = P[p +1 \dd c_i]$, and hence $p$ is a CT-border-period of $P_{i}$.
    Finally, since $p$ divides $\absolute{P}$, it also divides $\absolute{P_i} = \absolute{P} - (i-1)p$.
    This completes the induction and the statement follows.
\end{proof}

\begin{fact}[analogue of \cref{fact:pd_fragment}]\label{fact:nd_fragment}
Consider a string $X$.
The $\ND$ representation of a fragment $X[i \dd j]$ of $X$ is given by
\[
\ND(X[i \dd j])[t] =
\begin{cases}
  0, & \text{if } \ND(X)[i+t-1] > j-i-t+1, \\[6pt]
  \ND(X)[i+t-1], & \text{otherwise}.
\end{cases}
\] 
If $X[j]$ is the leftmost minimum in $X[i \dd j]$, then $\ND(X[i \dd j))$ is a prefix of $\ND(X)[i\dd j])$.
\end{fact}

\begin{restatable}[analogue of \cref{lem:no_overlaps}]{lemma}{noovstwo}\label{lem:no_overlaps_v2}
Consider a string $X$ and a fragment $X[i \dd j]$ with leftmost minimum $X[j]$.
Then, $p$ is a CT-block-period of $X[i\dd j]$ if and only if $p$ is a period of $\ND(X)[i\dd j)$ and $p$ divides $j-i+1$.
\end{restatable}
\begin{proof}
($\Rightarrow$:) By the definition of a CT-block-period, we have that $p$ divides $j-i+1$ and that
$\CT(X[i \dd j-p]) = \CT(X[i+p \dd j])$.
By \cref{lem:nsv_equiv}, this equality implies that
\begin{equation}\label{eq:simple2}
\ND(X[i \dd j-p]) = \ND(X[i+p \dd j]).
\end{equation}
Further,  by \cref{fact:nd_fragment}, since $X[j]$ and $X[j-p]$ are the minima of $X[i \dd j]$ and $X[i \dd j-p]$, respectively, we have that
$\ND(X[i \dd j))$ is a prefix of $\ND(X)[i \dd j)$   and that $\ND(X[i \dd j-p))$
is a prefix of $\ND(X)[1 \dd j-p)$.
By \eqref{eq:simple2}, we then have $\ND(X)[i \dd j-p) = \ND(X)[i+p \dd j)$
and hence $p$ is a period of $\ND(X)[i \dd j]$.

($\Leftarrow$:)
We first argue that it suffices to show that $X[j-p]$ is the leftmost minimum of $X[i \dd j-p]$.
If this is indeed the case, then by \cref{fact:nd_fragment} we have that $\ND(X)[i\dd j-p)=\ND(X)[i+p \dd j)$ is a prefix of $\ND(X[i \dd j-p])$, which, together with the fact that $\ND(X[i \dd j-p])[j-p-i+1] = 0 = \ND(X[i+p \dd j])[j-i-p+1]$, implies that $\ND(X[i \dd j-p]) = \ND(X[i+p \dd j])$.

Now, suppose for the sake of contradiction that the leftmost minimum of $X[i \dd j-p]$ is at some position $z \in [i, j-p]$.  
Then, $\nsv_X(z) = \ND(X)[z] +z > j-p$.
Since $\ND(X[i \dd j])$ is periodic with period $p$, we have 
$\ND(X)[z] = \ND(X)[z+p]$.
Therefore,
our assumption implies that $\nsv_X(z+p) =  \ND(X)[z+p]+(z+p) > j$,
which contradicts the assumption that $X[j]$ is the leftmost minimum of $X[i \dd j]$.
\end{proof}

\begin{lemma}[analogue of \cref{lem:no_sub_in_run}]\label{lem:no_sub_in_run:nsv}
    Consider $p, r, k \in \mathbb Z_{+}$ and strings $X = X_1UX_2$ and $Y = Y_1VY_2$ such that all of the following hold:
    \smallskip
    \begin{itemize}
        \item $\absolute{X_1} = \absolute{Y_1}$, $\absolute{X_2} = \absolute{Y_2}$, and $\absolute{U} = \absolute{V} = rp$, where $r \geq 2k + 1$;
        \item $U \approx V$;
        \item there is no position $i \in [1, \absolute{Y_1}]$ such that $\nsv_Y(i) \in [\absolute{Y_1} + p + 1, \absolute{Y_1V}]$;
        \item if we define $\forall i \in [1, r] : c_i := \absolute{X_1} + ip$, then $X[c_i]$ is the  leftmost minimum of $X(c_1 - p\dd c_i]$, and $Y[c_i]$ is the  leftmost minimum of $Y(c_1 - p\dd c_i]$;
        \item $\hdist(X \leadsto Y) \leq k$.
    \end{itemize}
    \smallskip
    Any sequence of $\hdist(X \leadsto Y)$ substitutions that transform $X$ into $X'$ with $X' \approx Y$ does not perform any substitution in fragment $X[c_{k + 1} \dd c_{r - k}]$.
\end{lemma}

\begin{proof}
    Consider any sequence of $\hdist(X \leadsto Y)$ substitutions that transforms $X$ into $X'$ such that ${X' \approx Y}$.
    We fix $i,j \in [1,r]$ to be the minimal and maximal values, respectively, such that $X[c_{i}]$ and $X[c_{j}]$ are not modified. (There are at least two distinct values due to $r \geq \hdist(X \leadsto Y) + 2$.)
    We will show that no substitution is performed in $X[c_i\dd c_j]$.

    Let $Z = X'[1 \dd c_i) \cdot X[c_i \dd c_j] \cdot X'(c_j \dd \absolute{X}] = X'[1 \dd c_i] \cdot X(c_i \dd c_j) \cdot X'[c_j \dd \absolute{X}]$. We claim that $Z \approx Y$.
    If this is not the case, there exists $x \in [1, \absolute{X}]$ with $\nsv_Z(x) \neq \nsv_Y(x)$. We show that such $x$ does not exist, distinguishing between the following four cases:

    \begin{description}
        \item[Case 1:] $x \in [1, c_1-p]$.
        By the assumptions, we have $\nsv_Y(x) \notin [\absolute{Y_1} + p + 1, \absolute{Y_1V}] = (c_1, c_r]$. 
        Now, suppose that $\nsv_Y(x) \leq c_1$.
        Note that $c_i \geq c_1$ implies $Z[1\dd c_1] = X'[1\dd c_1]$ and
        hence $\nsv_Z(x) = \nsv_{X'}(x) = \nsv_Y(x)$ follows from $X' \approx Y$, together with \cref{fact:subeq,lem:nsv_equiv}.
        If, however, $\nsv_{X'}(x) = \nsv_Y(x) > c_j \geq c_r$, then clearly $X'[x] \leq X'[c_j] = X[c_j]$.
        Since $X[c_j] = Z[c_j]$ is the  leftmost minimum of $X(c_i\dd c_j] = Z(c_i\dd c_j]$, it cannot be that $\nsv_Z(x) \in [c_i, c_j]$, which implies $\nsv_Z(x) = \nsv_{X'}(x) = \nsv_Y(x)$.
        \item[Case 2:] $x \in (c_1-p,c_i)$.
        By our assumptions, $Y[c_r]$ is the  leftmost minimum of $Y(c_1-p\dd c_r]$, and thus $\nsv_Y(x) \leq c_r$. Therefore, $\nsv_Z(x) = \nsv_{X'}(x) = \nsv_Y(x)$ follows from
        $Z[c_1-p\dd c_i] = X'[c_1-p\dd c_i] \approx Y[c_1-p\dd c_i]$.
        \item[Case 3:] $x \in [c_i, c_j)$.
        By our assumptions, $Y[c_j]$ is the  leftmost minimum of $Y[c_i \dd c_j]$.
        Hence, $\nsv_Y(x) \leq c_j$. We obtain $\nsv_Z(x) = \nsv_Y(x)$ due to \cref{lem:nsv_equiv} and the observation that
        \begin{alignat*}{1}
        Z[c_i \dd c_j] = X[c_i \dd c_j] = {}& U(ip \dd jp]\\
        {}\approx {}& V(ip \dd jp] = Y[c_i \dd c_j].
        \end{alignat*}%
        \item[Case 4:] $x \in [c_j, \absolute{X}]$.
        Since $Z$ has suffix $X'[c_j \dd \absolute{X}]$, and due to $X' \approx Y$, \cref{fact:subeq}, and \cref{lem:nsv_equiv}, it holds $\forall x \in [c_j, \absolute{X}] : \nsv_Z(x) = \nsv_{X'}(x) = \nsv_Y(x)$.
    \end{description}

    We have shown that $Z \approx Y$. Hence no optimal sequence of substitutions modifies $X[c_i\dd c_j]$. Finally, by the definition of $i$ and $j$, all of $c_1, \dots, c_{i - 1}$ and $c_{j + 1}, \dots, c_r$ were modified by our fixed sequence of substitutions.
    We therefore have $(i - 1) + (r - j) \leq \hdist(X \leadsto Y) \leq k$, which implies $i \leq k+1$ and $j \geq r - k$.
    Consequently, there are no substitutions in $X[c_{k + 1} \dd c_{r - k}]$.
\end{proof}

\begingroup
\def\ell{a}
\def\err{b}
\begin{lemma}[analogue of \cref{lem:delete_retain_ct_eq}]\label{lem:delete_retain_ct_eq:nsv}
    Consider two strings $X$ and $Y$ such that $X \approx Y$.
    Let $\ell, \err \in [1, \absolute{Y}]$ with $\ell \leq \err$ be chosen such that $Y[\err]$ is a leftmost minimum in $Y[\ell\dd \err]$.
    Then $X[1\dd \ell) \cdot X[\err\dd \absolute{X}] \approx Y[1\dd \ell) \cdot Y[\err\dd \absolute{X}]$.
\end{lemma}

\begin{proof}
    Due to $X \approx Y$, together with \cref{fact:subeq,lem:nsv_equiv}, it is clear that $X[\err]$ is minimal in $X[\ell\dd \err]$.
    For $Z \in \{X, Y\}$, let $\hat Z = Z[1\dd \ell) \cdot Z[\err\dd \absolute{Z}]$. It is easy to see that the following hold:
    \begin{itemize}
        \item For $i \in [\ell, \absolute{\hat Z})$, it holds $\nsv_{\hat Z}(i) = \nsv_{Z}(i) - \err + \ell$. (This is trivial.)
        \item For $i \in [1, \ell)$ with $\nsv_Z(i) \in [1, \ell)$, it holds $\nsv_{\hat Z}(i) = \nsv_Z(i)$. 
        \item For $i \in [1, \ell)$ with $\nsv_Z(i) \in [\ell, \err]$, it holds $\nsv_{\hat Z}(i) = \ell$. 
        \item For $i \in [1, \ell)$ with $\nsv_Z(i) \in (\err, \absolute{Z}]$, it holds $\nsv_{\hat Z}(i) = \nsv_Z(i) - \err + \ell$. 
    \end{itemize}
    Note that the definition of $\nsv_{\hat Z}$ depends only on $\nsv_Z$, and we already know that $\nsv_Z = \nsv_X = \nsv_Y$ due to $X \approx Y$ and \cref{lem:nsv_equiv}. Therefore, we have $\nsv_{\hat X} = \nsv_{\hat Y}$, and due to \cref{lem:nsv_equiv} this implies $\hat X \approx \hat Y$.
\end{proof}
\endgroup

\restatetrimnsv*

\begin{proof}
    Together, the following two statements are equivalent to the claim:
    
    \begin{claim}\label{claim:st1:nsv}
        $\hdist(X \leadsto Y) \leq k$ implies $\hdist(\hat X \leadsto \hat Y) \leq \hdist(X \leadsto Y)$.
    \end{claim}
    \begin{claim}\label{claim:st2:nsv}
        $\hdist(\hat X \leadsto \hat Y) \leq k$ implies $\hdist(\hat X \leadsto \hat Y) \geq \hdist(X \leadsto Y)$.
    \end{claim}
      
    \begin{claimproof}[Proof of \cref{claim:st1:nsv}]
    Assume that $\hdist(X \leadsto Y) \leq k$ and consider any sequence of $\hdist(X \leadsto Y)$ substitutions that transform $X$ into a string $X'$ with $X' \approx Y$.
    By \cref{lem:no_sub_in_run:nsv}, it holds $X' = X'[1\dd c_{k + 1}) \cdot X[c_{k + 1}\dd c_{r - k}] \cdot X'(c_{r - k} \dd \absolute{X}]$. 
    Note that we can perform the same sequence of substitutions (with an appropriate shift) on $\hat X$ as well.
    Hence, there is a sequence of $\hdist(X \leadsto Y)$ substitutions that transforms $\hat X$ into 
    \[\hat X' : = X'[1\dd c_{k + 1}) \cdot X[c_{k + 1}\dd c_{k + 1} + p) \cdot X[c_{r - k}] \cdot X'(c_{r - k} \dd \absolute{X}].\]%
    It remains to be shown that $\hat X' \approx \hat Y$, which implies $\hdist(\hat X \leadsto \hat Y) \leq \hdist(X \leadsto Y)$.
    Note that $\hat X'$ and $\hat Y$ can be obtained by deleting respectively fragments $X'[c_{k + 1} + p \dd c_{r - k})$ and $Y[c_{k + 1} + p\dd c_{r - k})$ from $X'$ and $Y$.
    By the definition of $Y$, we know that $Y[c_{r - k}]$ is a leftmost minimum in $Y[c_{k + 1} + p \dd c_{r - k}]$. Hence, we can apply \cref{lem:delete_retain_ct_eq:nsv} to $X'$ and $Y$ with $a = c_{k + 1} + p$ and $b = c_{r - k}$, thus obtaining $\hat X'\approx \hat Y$.
    \end{claimproof}

    \begin{claimproof}[Proof of \cref{claim:st2:nsv}]
    Assume that $\hdist(\hat X \leadsto \hat Y) \leq k$ and consider any sequence of $\hdist(\hat X \leadsto \hat Y)$ substitutions that transform $\hat X$ into a string $\hat X'$ with $
    \hat X' \approx \hat Y$.
    It can be readily verified that the strings $\hat X$ and $\hat Y$ still satisfy the conditions of \cref{lem:no_sub_in_run:nsv} with parameter $\hat r = 2k + 2$.
    Therefore, it holds 
    \begin{alignat*}{1}
        \hat X' = {}& \hat X'[1\dd c_{k + 1}) \cdot \hat X[c_{k + 1}\dd c_{\hat r - k}] \cdot \hat X'(c_{\hat r - k} \dd \absolute{\hat X}]\\
        {} = {}&\hat X'[1\dd c_{k + 1}) \cdot X[c_{k + 1}\dd c_{k + 1} + p) \cdot X[c_{r - k}] \cdot \hat X'(c_{k + 2} \dd \absolute{\hat X}].
    \end{alignat*}

    We observe that the same sequence of substitutions can also be applied directly to $X$, resulting in the string
    \[ X' = \hat X'[1\dd c_{k + 1}) \cdot X[c_{k + 1}\dd c_{r - k}] \cdot \hat X'(c_{k + 2} \dd \absolute{\hat X}].\]
    Let $\ell := \absolute{X} - \absolute{\hat X}$.     It remains to be shown that $X' \approx Y$, which implies $\hdist(\hat X \leadsto \hat Y) \geq {\hdist(X \leadsto Y)}$. Note that $\hat X'$ and $X'$, and also $\hat Y$ and $Y$ share a suffix of length $\absolute{X} - c_{r - k} + 1$. Because of ${\hat X' \approx \hat Y}$, combined with \cref{fact:subeq,lem:nsv_equiv}, it is clear that, for every $x \in (c_{r - k}, \absolute{X}]$, it holds $\nsv_{X'}(x) = \nsv_{\hat X'}(x - \ell) = \nsv_{\hat Y}(x - \ell) = \nsv_Y(x)$.
    
    Next, we observe that $X'[c_{k + 1}\dd c_{r - k}] = X[c_{k + 1}\dd c_{r - k}]$. Recall that $X[c_{r - k}]$ and $Y[c_{r - k}]$ are respectively minimal in $X[c_{k + 1}\dd c_{r - k}]$ and $Y[c_{k + 1}\dd c_{r - k}]$.
    Due to $U \approx V$, combined with \cref{fact:subeq,lem:nsv_equiv}, we have $\forall x \in [c_{k + 1}, c_{r - k}) : \nsv_{X'}(x) = \nsv_Y(x) \leq c_{r - k}$.
    
    Finally, we have to consider $x \in [1, c_{k + 1})$.
    We observe that, $\hat X'$ and $X'$, and also $\hat Y$ and~$Y$ share a prefix of length $c_{k + 1} + p - 1$.
    Hence, if $\nsv_Y(x) \leq c_{k + 1} + p - 1$, then $\hat X' \approx \hat Y$ already implies $\nsv_{\hat X'}(x) = \nsv_{\hat Y}(x)$.
    It remains to consider the case when $x \in [1, c_{k + 1})$ and $\nsv_Y(x) \geq c_{k + 1} + p$.
    Since $Y[c_{k +1}]$ is minimal in $Y(c_1 - p \dd c_{k + 1}]$, it holds $x \leq c_1 - p = \absolute{Y_1}$, and, by the condition of the lemma, $\nsv_Y(x) \geq \absolute{Y_1V} + 1 = c_{r} + 1$.
    Note that $\nsv_{\hat Y}(x) = \nsv_Y(x) - \ell \geq c_{r} + 1 - \ell = c_r + 1 - c_{r - k} + c_{k + 1} + p = c_{2k + 2} + 1$, as we obtained $\hat Y$ by merely deleting $\ell$ elements between $x$ and $\nsv_Y(x)$.
    Due to $\hat X' \approx \hat Y$, we have $\nsv_{\hat X'}(x) = \nsv_{\hat Y}(x)$.
    Finally, we obtain $X'$ from $\hat X'$ by inserting in $\hat X'$, between positions $x$ and $\nsv_{\hat X'}(x)$, values that are at least as large as $\hat X'[c_{k + 1} + p] = X[c_{r - k}]$. Due to $\hat X'[c_{k + 1} + p] \in (x, \nsv_{\hat X'}(x))$, we know that $\hat X'[x] < \hat X'[c_{k + 1} + p]$, and thus all the inserted values are at least $\hat X'[x]$.
    Consequently, we have $\nsv_{X'}(x) = \nsv_{\hat X'}(x) + \ell = \nsv_{\hat Y}(x) + \ell = \nsv_Y(x)$, as required.
    \end{claimproof}
    
    The combination of the two claims yields that either both $\hdist(\hat X \leadsto \hat Y)$ and $\hdist(X \leadsto Y)$ are at most $k$ and are in fact equal, or they are both greater than $k$.
\end{proof}

\section{Conclusions and Open Problems}

    We have presented an efficient algorithm for \problemname.
    Our algorithm can be seen as a reduction to several instances of computing $\hdist(X \leadsto Y)$ for pairs of strings $X$ and $Y$.
    We note that any improvement upon the dynamic-programming approach of Kim and Han \cite{DBLP:journals/tcs/KimH25} 
    (see also \cref{app:compare}) 
    for this comparison problem would directly yield an improvement over our main result.
    Interestingly, no (conditional) lower bound is known for either the problem of computing $\hdist(X \leadsto Y)$ or \problemname, and devising such a lower bound is another direction for future work.
    Finally, we hope that our techniques will serve as a foundation for advancing the state of the art in approximate CT-matching under the edit distance.

\bibliographystyle{plainurl}
\bibliography{references}

\clearpage

\appendix

\section{Efficient Computation of Cartesian Tree Hamming Distance}\label{app:compare}

First, we formally define a convolution problem on weakly monotone sequences and present a self-contained version of the efficient solution of Kim and Han~\cite{DBLP:journals/tcs/KimH25}.

\defproblem{\textsc{Max-Min Convolution for Non-Decreasing Sequences}}
{Two non-decreasing sequences $A[i]_{i=0}^{i=z}$ and $B[i]_{i=0}^{i=z}$.}
{A sequence $C[i]_{i=0}^{i=z}$ such that $C[i]= \max\limits_{j \in [0,i]}(\min \{A[j], B[i-j]\})$.}

\begin{lemma}\label{lem:conv}
    The $\textsc{Max-Min Convolution for Non-Decreasing Sequences}$ problem can be solved in $\cO(z)$ time.
\end{lemma}
\begin{proof}
    We define, for each $i \in [0,z]$, breakpoint $\mathit{BP}[i]$ as the largest integer $j \in [-1, i]$ for which
    $A[j] \leq B[i-j]$.
    Note that $\mathit{BP}$ is non-decreasing due to the monotonicity of $A$ and $B$.
    
    In what follows, we first describe how array $\mathit{BP}$ enables an efficient computation of array~$C$ and then present an $\cO(z)$-time algorithm for constructing $\mathit{BP}$.
    
	For convenience, we set $A[-1] := -\infty$ and $B[-1] = -\infty$.
    \begin{claim}
    	For all $i \in [0,z]$, we have
    $C[i]= \max\{ A[\mathit{BP}[i]], B[i-\mathit{BP}[i]-1]\}$.
    \end{claim}
	\begin{claimproof}
	We can rewrite
	$C[i]= \max\limits_{j \in [0,i]}(\min \{A[j], B[i-j]\})$ as
	\[C[i]= \max \, \left\{ \, \max\limits_{j \in [0,\mathit{BP}[i]]} A[j] \,\, ,  \,\, \max\limits_{j \in (\mathit{BP}[i],i+1]} B[i-j] \, \right\} .\]
	The claim then follows by the monotonicity of arrays $A$ and $B$.
	\end{claimproof}
    The above claim implies that given array $\mathit{BP}$, we can compute $C[i]$ for any $i \in [0,z]$ in $\cO(1)$ time, and we can thus compute $C$ in $\cO(z)$ time in total.
    
    It remains to show that array $\mathit{BP}$ can be computed in $\cO(z)$ time given $A$ and $B$.
    We first compute an auxiliary array $Z[0\dd z]$, where $Z[j]$ is the number of the elements of~$B$ that are strictly smaller than $A[j]$, noting that $Z$ is a non-decreasing sequence since $A$ is non-decreasing.
    Array $Z$ can be computed in $\cO(z)$ time by merging $A$ and~$B$ in a merge-sort fashion.
    
    We then initialise $\mathit{BP}[0\dd z]$ as an array all of whose entries are set to $-1$ and, intuitively, for each $j$ from $0$ to $z$, we then want to increment all entries in
    $BP[j + Z[j] \dd z]$ by $1$.
    We can do this in $\cO(z)$ total time using a left-to-right sweep to accumulate said increments, effectively computing prefix sums.
    This concludes the proof of the lemma.
\end{proof}

\compare*
\begin{proof}
Consider a node $v$ of $\CT(P)$ and an integer $x \in [0, k]$.
Let the subtree of $\CT(P)$ rooted at $v$ correspond to some fragment $P_v = P[i \dd i+\absolute{str(v)})$ with leftmost minimum at position $P[j]$.
Let $\mathcal{Y}(v,x)$ denote the set of strings $Y$ that can be obtained from $T[i \dd i+\absolute{str(v)})$ with at most $x$ substitutions that satisfy $Y \approx P_v$.
We define the following three values:
\begin{align*}
    \dphd(v, x) & := \max\limits_{Y \in \mathcal{Y}(v,x)\cup \{-\infty\}} \min Y, \\
    \dphdsub(v, x) & :=  \max\limits_{Y \in \{Y \in \mathcal{Y}(v,x) \mid Y[j]\neq T[j]\} \cup \{-\infty\}} \min Y, \\
    \dphdnosub(v, x) & := \max\limits_{Y \in \{Y \in \mathcal{Y}(v,x) \mid Y[j]= T[j]\} \cup \{-\infty\}} \min Y.
\end{align*}
\cref{fig:ct-computation} illustrates tables $\dphdsub$, $\dphdnosub$, and $\dphd$.

Following \cite{DBLP:journals/tcs/KimH25}, we next provide recursive formulas that allow us to construct the tables defined above; see \cref{fig:ct-computation}. 
\begin{itemize}
    \item Let $v$ be a leaf node corresponding to $P[i]$.
    For $x \geq 1$ we can substitute the $T[i]$ with an arbitrarily large and hence have
    \begin{alignat}{1}
    & \dphd(v, 0) = T[i]\textnormal, \label{eq:in1} \\
    & \dphd(v, x) = \infty \textnormal{\ for\ } x \in [1, k]. \label{eq:in2}
\end{alignat}

    \item Let $v$ be an internal node with children $\ell$ and $r$, let $\mu:= T[i+\absolute{str(\ell)}]$.
    As the value of the root node $v$ must be strictly smaller than the value of its left child $\ell$ and smaller than or equal to the value of its right child $r$, we have, for $x \in [0, k]$:
    \begin{alignat}{1}
    \dphdsub(v, x) =\ & \begin{cases}
    \quad\ -\infty & \textnormal{if $x = 0$}\\
    \ \ \max\limits_{\overset{\scriptstyle y,z \in \mathbb Z_{\geq 0}}{\mathclap{y + z = x - 1}}}(\min \{ \dphd(\ell, y) - 1, \dphd(r, z)\}) & \textnormal{if $x > 0$}\end{cases} \label{eq:opthsub}\\[.5\baselineskip]
    \dphdnosub(v, x) =\ & \color{white}\begin{cases}\end{cases}\color{black}\ \ \max\limits_{\overset{\scriptstyle y,z \in \mathbb Z_{\geq 0}}{y + z = x}}(\min\{\dphd(\ell, y)-1, \dphd(r, z)\}). \label{eq:opthnosub}
\end{alignat}

Then, to compute $\dphd(v, x)$ from $\dphdsub(v, x)$ and $\dphdnosub(v, x)$, it suffices to note that
if $\mu \leq \dphdnosub(v, x)$ we can afford to not perform a substitution at $T[i+\absolute{str(\ell)}]$.
In this case, we can hence take the maximum of $\mu$ and the maximum possible value obtain by performing a substitution ($\dphdsub(v, x)$).
Otherwise, substituting $T[i+\absolute{str(\ell)}]$ is mandatory.
We thus have:
\begin{alignat}{1}
    \dphd(v, x) =\ & \begin{cases}
        \ \max\{\dphdsub(v, x), \mu \} & \textnormal{if $\mu \leq \dphdnosub(v, x) $}\\
        \phantom{\ \max\{ }\dphdsub(v, x) & \textnormal{otherwise.}
    \end{cases}\label{eq:edit}
\end{alignat}
\end{itemize}

It is easy to see that $\hdist_k(T \leadsto P) = \min( \{x \in [0, k] \mid \dphd(r, x) \neq -\infty \} \cup \{k+1\})$, where $r$ is the root of $\CT(P)$.
It remains to describe how we compute the tables $\dphd, \dphdsub$ and $\dphdnosub$.
The computations are performed in the order of increasing subtree sizes (resolving tries arbitrarily), so that when we want to compute the entries for a node $v$, all the entries for its children $\ell$ and $r$ have been already computed.

We now describe how to populate $\dphdsub(\cdot, \cdot)$; the computations for $\dphdnosub(\cdot, \cdot)$ are analogous.
At an internal node $v$ with children $\ell$ and $r$, the algorithm has already computed entries $\dphd(\ell, x)$ and $\dphd(r, x)$ for $x\in [0,k]$.
Define $A[x] = \dphd(\ell,x)-1$ and $B[x] = \dphd(r ,x)$ for $x \in [0,k]$.
Observe that $A$ and $B$ are non-decreasing by the definition of $\dphd$.
Let $C$ be the output of an application of \cref{lem:conv} to $A$ and $B$.
Then, for $x \in (0,k]$, we have $C[x] = \dphdsub(v,x)$.
Thus, it suffices to invoke \cref{lem:conv} once for each node of $\CT(P)$  to populate $\dphdsub(\cdot, \cdot)$.
The total cost of all such invocations is $\cO(mk)$ and this is also an upper bound on the time required to populate $\dphdnosub(\cdot, \cdot)$.
\end{proof}

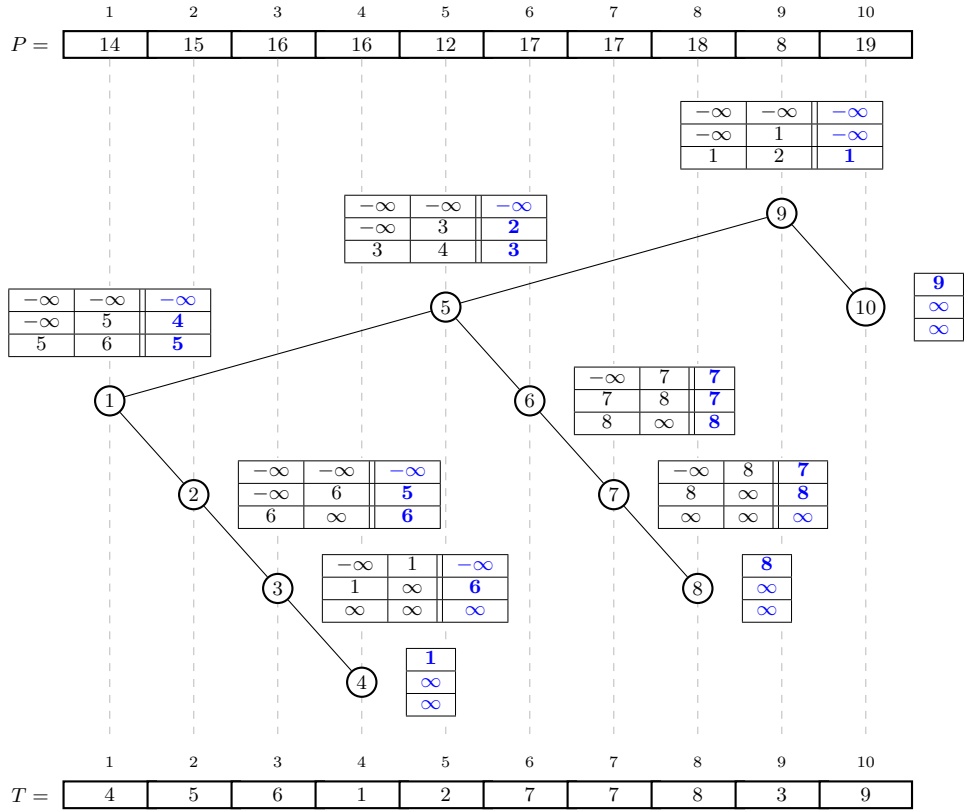
\begin{figure}[t]
\centering
{\input{figs/edit}}
\caption{The Cartesian tree $\CT(P)$ of the pattern $P = [14, 15, 16, 16, 12, 17, 17, 18, 8, 19]$, where each node is labeled by the index of the value it represents.
We consider the text $T = [4, 5, 6, 1, 2, 7, 7, 8, 3, 9]$ which satisfies $T \not\approx P$.
However, we have $\hdist(T \leadsto P) = 2$; by substituting the values at positions 4 and 9, we obtain $T' = [4, 5, 6, {\color{red}\mathbf{\infty}}, 2, 7, 7, 8, {\color{red}\mathbf{1}}, 9] \approx P$.
For each leaf node $v$ and each integer $x \in [0,2]$, the shown table values $\dphd(v,x)$ are initialized using \eqref{eq:in1} and \eqref{eq:in2} (shown in blue).
For internal nodes, the tables display $\dphdsub(v,x)$ (in Column 1) and $\dphdnosub(v,x)$ (in Column 2) computed via \eqref{eq:opthsub} and \eqref{eq:opthnosub}, respectively, as well as $\dphd(v,x)$ (in Column 3 in blue), computed according to \eqref{eq:edit}.}
\label{fig:ct-computation}
\end{figure}

\end{document}

%% file: figs/plot.tex



\begin{tikzpicture}

    \pgfmathsetmacro{\sx}{5}
    \pgfmathsetmacro{\sy}{4}

    \node[fill,circle, inner sep=.75pt] (o) at (0,0) {};
    \node[fill,circle, inner sep=1pt] (1) at (\sx*0.75,0) {};
    \node[fill,circle, inner sep=1pt] (2) at  (\sx*0.9375,0) {};
    \node[fill,circle, inner sep=1pt] (4) at  (\sx*0.75,\sy*1) {};
    \node[fill,circle, inner sep=.75pt] (4a) at  (\sx*0.75,0) {};
    \node[fill,circle, inner sep=.75pt] (4b) at  (0,\sy*1) {};
    \node[fill,circle, inner sep=1pt] (5) at  (\sx*0.9375,\sy*1.2) {};
    \node[fill,circle, inner sep=.75pt] (5a) at  (\sx*0.9375,0) {};
    \node[fill,circle, inner sep=.75pt] (5b) at  (0,\sy*1.2) {};
    \node[fill,circle, inner sep=1pt] (6) at  (\sx*2.25,\sy*1.515
) {};
    \coordinate (6p) at (\sx*3,\sy*2) {};
    \node[fill,circle, inner sep=.75pt] (6a) at  (\sx*2.25,0) {};
    \coordinate (6ap) at  (\sx*3,0) {};
    \node[fill,circle, inner sep=.75pt] (6b) at  (0,\sy*1.515
) {};
    \coordinate (6bp) at  (0,\sy*2) {};
    
    \node[fill,circle, inner sep=1pt] (9) at  (0,\sy*1) {};
    \node[fill,circle, inner sep=1pt] (10) at  (0,\sy*.5) {};

    \coordinate (62) at (\sx*1.8,\sy*1.40625
) {};
    \coordinate (622) at (\sx*1.86,\sy*1.42025
) {};
    \coordinate (92) at (\sx*1.8,\sy*.6) {};
    \coordinate (922) at (\sx*1.86,\sy*1.64) {};

    \coordinate (7a) at  (\sx*1.795,0) {};
    \coordinate (7a2) at  (\sx*1.865,0) {};
    \coordinate (7a') at  (\sx*1.8,0) {};
    \coordinate (7a2') at  (\sx*1.86,0) {};
    \coordinate (7b) at  (0,\sy*1.39625) {};
    \coordinate (7b2) at  (0,\sy*1.43025
) {};
    \coordinate (7b') at  (0,\sy*1.40625
) {};
    \coordinate (7b2') at  (0,\sy*1.42025
) {};

    \draw[very thin,black!10] (7a') -- (62) -- (7b');
    \draw[very thin,black!10] (7a2') -- (622) -- (7b2');

    \draw[thin] (62) -- ++(180:.15);
    \draw[thin] (62) -- ++(90:-.15);
    \draw[thin] (622) -- ++(180:.15);
    \draw[thin] (622) -- ++(90:-.15);


    \draw[very thin,black!30] (6a) -- (6) -- (6b);
    \draw[very thin,black!30] (5a) -- (5) -- (5b);
    \draw[very thin,black!30] (4a) -- (4) -- (4b);

    \draw (o) -- (6a);
    \draw (o) -- (6b);
    \draw[white,thick] (7b) -- (7b2);
    \draw[thin] (7b) -- ++(45:.3);
    \draw[thin] (7b) -- ++(45:-.3);
    \draw[thin] (7b2) -- ++(45:.3);
    \draw[thin] (7b2) -- ++(45:-.3);

    \draw[white,thick] (7a) -- (7a2);
    \draw[thin] (7a) -- ++(45:.3);
    \draw[thin] (7a) -- ++(45:-.3);
    \draw[thin] (7a2) -- ++(45:.3);
    \draw[thin] (7a2) -- ++(45:-.3);

    \draw[color=purple!90!black, line width=2pt] (10) -- (4) 
    node[pos=.45, below=.25em,
    sloped]{
    This work:\ \(\cO(n \sqrt{n} \cdot k^{2.5})\)}; 
    \draw[color=purple!90!black, line width=2pt] (4) -- (5) node[pos=.4, below=.25em,
    sloped,inner sep=0.5ex, fill=white]{
    \(\cO(n k^{5})\)}; 
    \draw[color=black, line width=.5pt] (622) -- (6) (9) -- (5) (5) -- (62) node[midway, below=.25em, sloped]
    {\(\cO(n^2k)\), Kim and Han~\cite{DBLP:journals/tcs/KimH25}}; 


    \def\slfrac#1#2{\ensuremath{}^{#1}\!/\!{}_{#2}}

    \node[anchor=north east,outer sep=.1em,rotate=45] at (o)  { $k \approx 1$};
    \node[anchor=north east,outer sep=.1em,rotate=45] at (1)  { $k \approx n^{1/5}$};
    \node[anchor=north east,outer sep=.1em,rotate=45] at (2)  { $k \approx n^{1/4}$};
    \node[anchor=north east,outer sep=.1em,rotate=45] at (6a)  { $k \approx n$};

    \node[anchor=east,outer sep=.5em] at (o) { $t(n,k) \approx n$};
    \node[anchor=east,outer sep=.5em] at (10) { $t(n,k) \approx n^{3/2}$};
    \node[anchor=east,outer sep=.5em] at (4b) { $t(n,k) \approx n^{2}$};
    \node[anchor=east,outer sep=.5em] at (5b) { $t(n,k) \approx n^{9/4}$};
    \node[anchor=east,outer sep=.5em] at (6b) { $t(n,k) \approx n^3$};
\end{tikzpicture}


%% file: figs/ct-tree-tables.tex


\begingroup
\def\onex{1.35em}
\def\oney{1.5em}
\small
\begin{tikzpicture}[
    treenode/.style={circle, draw=black, fill=white, thick, 
    inner sep=2.5pt,
    text centered, text=black, 
    },
    arraynode/.style={rectangle, 
    draw=black, 
    fill=white,
    minimum width=1.2*\onex,
    text centered, 
    },
    arrow/.style={thick},
    y=\oney, x=\onex,
    baseline=(current bounding box.north),
]

\node[arraynode] (x1) at (0, 1) {4};
\node[font=\small, left=.25em of x1] {$X\,=$};
\node[arraynode] (x2) at (1.2, 1) {5};
\node[arraynode] (x3) at (2.4, 1) {6};
\node[arraynode] (x4) at (3.6, 1) {1};
\node[arraynode] (x5) at (4.8, 1) {2};
\node[arraynode] (x6) at (6, 1) {7};
\node[arraynode] (x7) at (7.2, 1) {7};
\node[arraynode] (x8) at (8.4, 1) {8};
\node[arraynode] (x9) at (9.6, 1) {3};
\node[arraynode] (x10) at (10.8, 1) {9};

\foreach \i in {1,...,10} {
  \node[above=.25em of x\i, inner sep=0] {$\scriptstyle \i$};
}

\foreach \i in {1,...,10} {
\draw[dashed, gray] (x\i.south) -- ++(0, -5.5);
}

\def\ystep{.8}

\begin{scope}[yshift = .7*\oney]
\node[treenode] (n1)  at (3.6,  -1)               {1};
\node[treenode] (n4)  at (0,    {-1-1*\ystep})    {4};
\node[treenode] (n2)  at (4.8,  {-1-1*\ystep})    {2};
\node[treenode] (n5)  at (1.2,  {-1-2*\ystep})    {5};
\node[treenode] (n3)  at (9.6,  {-1-2*\ystep})    {3};
\node[treenode] (n6)  at (2.4,  {-1-3*\ystep})    {6};
\node[treenode] (n7a) at (6,    {-1-3*\ystep})    {7};
\node[treenode] (n9)  at (10.8, {-1-3*\ystep})    {9};
\node[treenode] (n7b) at (7.2,  {-1-4*\ystep})    {7};
\node[treenode] (n8)  at (8.4,  {-1-5*\ystep})    {8};

\draw[thick] (n1) -- (n4);
\draw[thick] (n1) -- (n2);
\draw[thick] (n4) -- (n5);
\draw[thick] (n2) -- (n3);
\draw[thick] (n5) -- (n6);
\draw[thick] (n3) -- (n7a);
\draw[thick] (n3) -- (n9);
\draw[thick] (n7a) -- (n7b);
\draw[thick] (n7b) -- (n8);

\end{scope}

\begin{scope}[yshift = 2*\oney]
\node[arraynode] (z1)  at (0,   -8) {\clap{14}};
\node[font=\small, left=.25em of z1] {$Z\,=$};
\node[arraynode] (z2)  at (1.2, -8) {\clap{15}};
\node[arraynode] (z3)  at (2.4, -8) {\clap{16}};
\node[arraynode] (z4)  at (3.6, -8) {\clap{16}};
\node[arraynode] (z5)  at (4.8, -8) {\clap{12}};
\node[arraynode] (z6)  at (6,   -8) {\clap{17}};
\node[arraynode] (z7)  at (7.2, -8) {\clap{17}};
\node[arraynode] (z8)  at (8.4, -8) {\clap{18}};
\node[arraynode] (z9)  at (9.6, -8) {\clap{8}};
\node[arraynode] (z10) at (10.8,-8) {\clap{19}};


\foreach \i in {1,...,10} {
\draw[dashed, gray] (z\i.south) -- ++(0, -5.5);
}
\begin{scope}[yshift = -8.3*\oney]
\node[treenode] (m8)   at (9.6,  -1)               {8};

\node[treenode] (m12)  at (4.8,  {-1-1*\ystep})    {\phantom{0}};
\node[treenode,draw=none,fill=none]                at (4.8,  {-1-1*\ystep}) {12};

\node[treenode] (m19)  at (10.8, {-1-1*\ystep})    {\phantom{0}};
\node[treenode,draw=none,fill=none]                at (10.8, {-1-1*\ystep}) {19};

\node[treenode] (m14)  at (0,    {-1-2*\ystep})    {\phantom{0}};
\node[treenode,draw=none,fill=none]                at (0,    {-1-2*\ystep}) {14};

\node[treenode] (m17a) at (6,    {-1-2*\ystep})    {\phantom{0}};
\node[treenode,draw=none,fill=none]                at (6,    {-1-2*\ystep}) {17};

\node[treenode] (m15)  at (1.2,  {-1-3*\ystep})    {\phantom{0}};
\node[treenode,draw=none,fill=none]                at (1.2,  {-1-3*\ystep}) {15};

\node[treenode] (m17b) at (7.2,  {-1-3*\ystep})    {\phantom{0}};
\node[treenode,draw=none,fill=none]                at (7.2,  {-1-3*\ystep}) {17};

\node[treenode] (m16a) at (2.4,  {-1-4*\ystep})    {\phantom{0}};
\node[treenode,draw=none,fill=none]                at (2.4,  {-1-4*\ystep}) {16};

\node[treenode] (m18)  at (8.4,  {-1-4*\ystep})    {\phantom{0}};
\node[treenode,draw=none,fill=none]                at (8.4,  {-1-4*\ystep}) {18};

\node[treenode] (m16b) at (3.6,  {-1-5*\ystep})    {\phantom{0}};
\node[treenode,draw=none,fill=none]                at (3.6,  {-1-5*\ystep}) {16};

\draw[thick] (m8)  -- (m12);
\draw[thick] (m8)  -- (m19);

\draw[thick] (m12) -- (m14);
\draw[thick] (m12) -- (m17a);

\draw[thick] (m14) -- (m15);
\draw[thick] (m15) -- (m16a);
\draw[thick] (m16a)-- (m16b);

\draw[thick] (m17a)-- (m17b);
\draw[thick] (m17b)-- (m18);
\end{scope}\end{scope}

\end{tikzpicture}~\hfill~%
\begin{tikzpicture}[
    treenode/.style={circle, draw=black, fill=white, thick, 
    inner sep=2pt,
    text centered, text=black, 
    },
    arraynode/.style={rectangle, 
    draw=black, 
    fill=white,
    minimum width=1.2*\onex,
    text centered, 
    },
    arrow/.style={thick},
    y=\oney, x=\onex,
    baseline=(current bounding box.north),
]

\node[arraynode] (x1) at (0, 1) {4};
\node[font=\small, left=.25em of x1] {$X\,=$};
\node[arraynode] (x2) at (1.2, 1) {5};
\node[arraynode] (x3) at (2.4, 1) {6};
\node[arraynode] (x4) at (3.6, 1) {1};
\node[arraynode] (x5) at (4.8, 1) {2};
\node[arraynode] (x6) at (6, 1) {7};
\node[arraynode] (x7) at (7.2, 1) {7};
\node[arraynode] (x8) at (8.4, 1) {8};
\node[arraynode] (x9) at (9.6, 1) {3};
\node[arraynode] (x10) at (10.8, 1) {9};

\foreach \i in {1,...,10} {
  \node[above=.25em of x\i, inner sep=0] {$\scriptstyle \i$};
}

\foreach \i in {1,...,10} {
\draw[dashed, gray] (x\i.south) -- ++(0, -11);
}

\begin{scope}[yshift = -.3*\oney]
\node[treenode] (r) at (-1.2, 0) {\phantom{0}};

\node[treenode] (n1)  at (0,   -1) {4};
\node[treenode] (n2)  at (1.2, -2) {5};
\node[treenode] (n3)  at (2.4, -3) {6};
\node[treenode] (n4)  at (3.6, -1) {1};
\node[treenode] (n5)  at (4.8, -2) {2};
\node[treenode] (n6)  at (6,   -3) {7};
\node[treenode] (n7)  at (7.2, -4) {7};
\node[treenode] (n8)  at (8.4, -5) {8};
\node[treenode] (n9)  at (9.6, -3) {3};
\node[treenode] (n10) at (10.8,-4) {9};

\draw[thick] (r)  -- (n1);
\draw[thick] (r)  -- (n4);

\draw[thick] (n1) -- (n2);
\draw[thick] (n2) -- (n3);

\draw[thick] (n4) -- (n5);
\draw[thick] (n5) -- (n6);
\draw[thick] (n6) -- (n7);
\draw[thick] (n7) -- (n8);

\draw[thick] (n5) -- (n9);
\draw[thick] (n9) -- (n10);
\end{scope}

\def\ypsv{-6.5}
\def\ypd{-7.5}
\def\ynsv{-9.5}
\def\ynd{-10.5}

\node[arraynode] (psv1)  at (0,   \ypsv) {0};
\node[font=\small, left=.25em of psv1] {$\psv_X\,=$};
\node[arraynode] (psv2)  at (1.2, \ypsv) {1};
\node[arraynode] (psv3)  at (2.4, \ypsv) {2};
\node[arraynode] (psv4)  at (3.6, \ypsv) {0};
\node[arraynode] (psv5)  at (4.8, \ypsv) {4};
\node[arraynode] (psv6)  at (6,   \ypsv) {5};
\node[arraynode] (psv7)  at (7.2, \ypsv) {6};
\node[arraynode] (psv8)  at (8.4, \ypsv) {7};
\node[arraynode] (psv9)  at (9.6, \ypsv) {5};
\node[arraynode] (psv10) at (10.8,\ypsv) {9};

\node[arraynode] (pd1)  at (0,   \ypd) {0};
\node[font=\small, left=.25em of pd1] {$\PD(X)\,=$};
\node[arraynode] (pd2)  at (1.2, \ypd) {1};
\node[arraynode] (pd3)  at (2.4, \ypd) {1};
\node[arraynode] (pd4)  at (3.6, \ypd) {0};
\node[arraynode] (pd5)  at (4.8, \ypd) {1};
\node[arraynode] (pd6)  at (6,   \ypd) {1};
\node[arraynode] (pd7)  at (7.2, \ypd) {1};
\node[arraynode] (pd8)  at (8.4, \ypd) {1};
\node[arraynode] (pd9)  at (9.6, \ypd) {4};
\node[arraynode] (pd10) at (10.8,\ypd) {1};

\node[arraynode] (nsv1)  at (0,   \ynsv) {4};
\node[font=\small, left=.25em of nsv1] {$\nsv_X\,=$};
\node[arraynode] (nsv2)  at (1.2, \ynsv) {4};
\node[arraynode] (nsv3)  at (2.4, \ynsv) {4};
\node[arraynode] (nsv4)  at (3.6, \ynsv) {11};
\node[arraynode] (nsv5)  at (4.8, \ynsv) {11};
\node[arraynode] (nsv6)  at (6,   \ynsv) {9};
\node[arraynode] (nsv7)  at (7.2, \ynsv) {9};
\node[arraynode] (nsv8)  at (8.4, \ynsv) {9};
\node[arraynode] (nsv9)  at (9.6, \ynsv) {11};
\node[arraynode] (nsv10) at (10.8,\ynsv) {11};

\node[arraynode] (nd1)  at (0,   \ynd) {3};
\node[font=\small, left=.25em of nd1] {$\ND(X)\,=$};
\node[arraynode] (nd2)  at (1.2, \ynd) {2};
\node[arraynode] (nd3)  at (2.4, \ynd) {1};
\node[arraynode] (nd4)  at (3.6, \ynd) {0};
\node[arraynode] (nd5)  at (4.8, \ynd) {0};
\node[arraynode] (nd6)  at (6,   \ynd) {3};
\node[arraynode] (nd7)  at (7.2, \ynd) {2};
\node[arraynode] (nd8)  at (8.4, \ynd) {1};
\node[arraynode] (nd9)  at (9.6, \ynd) {0};
\node[arraynode] (nd10) at (10.8,\ynd) {0};

\end{tikzpicture}
\endgroup


%% file: figs/blocks.tex
\begin{figure}
\centering
{\small

\edef\onex{1.25em}
\edef\oney{1.1em}

\begin{tikzpicture}[
    every node/.style={inner sep=0pt},
    x=\onex,
    y=\oney,
    every fit/.style={inner sep=-0.2pt}
]

\foreach[count=\y from 0] \ypos in {0,-2,-13,-15} {
    \foreach \i in {0,...,30} {
        \node[minimum width=\onex, minimum height=\oney] (\i-\y) at (\i,\ypos) {};
        \node[fit=(\i-\y), draw, gray] {};
    }
}

\node[fit=(9-0)(0-0),   draw=white, fill=white, inner sep=1pt] {};
\node[fit=(30-0)(26-0), draw=white, fill=white, inner sep=1pt] {};
\node[fit=(5-2)(0-2),   draw=white, fill=white, inner sep=1pt] {};
\node[fit=(30-2)(26-2), draw=white, fill=white, inner sep=1pt] {};

\foreach \x in {11-2, 11-3, 14-2, 14-3, 19-2, 19-3} {
\node[fit=(\x.north west)(\x.south west), draw=white, fill=white, inner xsep=1pt, inner ysep=-.4pt] (tmp) {};
\node at (tmp) {$\dots$};
}

\node[
    fit=(6-1)(0-1),
    draw=white,
    fill=white,
    inner sep=1pt
] {};

\node[
    fit=(5-3)(0-3),
    draw=white,
    fill=white,
    inner sep=1pt
] {};

\node[
    fit=(26-3)(30-3),
    draw=white,
    fill=white,
    inner sep=1pt
] {};

\foreach \i in {9,10,11,13,21,23} {
    \node[fit=(\i-0), fill=black!20!white, draw=gray] {};
}

\foreach \i in {7,8,10,13,21,23} {
    \node[fit=(\i-1), fill=black!20!white, draw=gray] {};
}

\foreach \i in {9,22} {
    \node[fit=(\i-2), fill=black!20!white, draw=gray] {};
}

\foreach \i in {8,20,21,24} {
    \node[
        fit=(\i-2),
        pattern=north east lines,
        pattern color=gray,
        draw=gray
    ] {};
}

\foreach \i in {12} {
    \node[fit=(\i-3), fill=black!20!white, draw=gray] {};
}

\node[fit=(6-0)(25-0), draw] (pcase1) {};
\node[fit=(0-1)(30-1), draw] (tcase1) {};
\node[fit=(0-3)(30-3), draw] (tcase2) {};

\node[left=0 of pcase1] (P1) {$P=\ $};
\node[left=0 of tcase1] (T1) {$T=\ $};
\node[left=0 of tcase2] (T2) {$T=\ $};

\node[fit=(13-0)(13-0), draw=red, line width=1.5pt, inner sep=.75pt] {};
\node[fit=(13-0)(13-0), draw=black, thin] (hlbox) {};
\node[fit=(13-1)(13-1), draw=red, line width=1.5pt, inner sep=.75pt] {};
\node[fit=(13-1)(13-1), draw=black, thin] (hlbox) {};

\node[
    below right=5em and 0 of T1.west
] (C1) {
    \parbox{\textwidth}{
        \textbf{\textsf{Case I:}}
        If enough bad blocks accumulate in the pattern,
        then, for every approximate occurrence, a bad block of the text must be aligned with
        a bad block of the pattern (e.g., the highlighted block).
    }
};

\node[
    below right=5em and 0 of T2.west
] (C2) {
    \parbox{\textwidth}{
        \textbf{\textsf{Case II:}}
        Of a maximal fragment of blocks that are good in both $P$ and $T$ (e.g., the one between the dotted lines), we can trim all but the initial and final $3k$ blocks (e.g., the highlighted three blocks).
    }
};

\draw[
    decorate,
    decoration={brace, amplitude=4pt},
    line width=1pt
]
($(16-0.north west)+(1pt,0.1em)$) --
($(18-0.north east)+(-1pt,0.1em)$)
node[midway, above=.6em] {$P[\ell..r)$};

\draw[
    decorate,
    decoration={brace, amplitude=4pt},
    line width=1pt
]
($(9-0.north west)+(1pt,0.1em)$) --
($(15-0.north east)+(-1pt,0.1em)$)
coordinate[midway, above=.5em] (badbrace);

\node[
    above right=3.5em and 1em of T1.north west,
    align=center,
    inner xsep=.5em
] (accum) {
    $3k+1$ bad blocks on\\one side, where $k = 1$
};

\draw[-Latex] (accum) to[out=10, in=90] (badbrace);

\draw[dotted]
    (8-0.north east) -- ++(0,2em) coordinate (vleft);

\draw[dotted]
    (25-0.north east) -- ++(0,2em) coordinate (vright);

\draw[
    decorate,
    decoration={brace, amplitude=4pt},
    line width=1pt
]
(vleft) -- (vright)
node[midway, above=.6em] {$\Phi$};

\draw[
    decorate,
    decoration={brace, amplitude=4pt, mirror},
    line width=1pt
]
($(14-1.south west)+(1pt,-0.1em)$) --
($(20-1.south east)+(-1pt,-0.1em)$)
node[midway, below=.6em] {CT-run};

\draw[dotted]
    (6-1.south east) -- ++(0,-2em) coordinate (vleft2);

\draw[dotted]
    (30-1.south east) -- ++(0,-2em) coordinate (vright2);

\draw[
    decorate,
    decoration={brace, amplitude=4pt, mirror},
    line width=1pt
]
(vleft2) -- (vright2)
node[midway, below=.6em]
{CT-run extended with up to $3k+1$ bad blocks in each direction};

\node[fit=(6-2),  draw=white, fill=white] {};
\node[fit=(25-2), draw=white, fill=white] {};

\node[
  minimum width=1.5*\onex,
  minimum height=\oney,
  fill=white,
  draw=gray,
  left=0 of 7-2.west,
] (6-2-wide) {};

\node[
  minimum width=1.5*\onex,
  minimum height=\oney,
  fill=white,
  draw=gray,
  right=0 of 24-2.east,
] (25-2-wide) {};

\node[fit=(6-2-wide)(25-2-wide), draw] (pcase2) {};
\node[left=0 of pcase2] (P2) {$P=\ $};

\node at (6-2-wide.center) {$X_P$};
\node at (25-2-wide.center) {$Y_P$};

\draw[
    decorate,
    decoration={brace, amplitude=4pt},
    line width=1pt
]
($(6-2-wide.north east)+(0pt,0.3em)$) --
($(25-2-wide.north west)+(0pt,0.3em)$)
node[midway, above=.6em]
{$B_P$ contains at most $6k$ explicitly bad blocks};

\foreach[evaluate=\i as \iminus using int(\i - 1)] \i in {14,19} {
  \draw[
    decorate,
    decoration={brace, amplitude=4pt, mirror},
    line width=1pt
  ]
  ($( \iminus-3.south west)+(1pt,-0.1em)$) --
  ($( \i-3.south east)+(-1pt,-0.1em)$) node[midway, below=.6em, align=center] {$3k$ good\\blocks};
  
}

\foreach \i in {13,20} {
\draw[line width=3pt, white] (\i-2.north west) ++(0,.25em) to (\i-3.south west);
\draw[very thick, densely dotted] (\i-2.north west) ++(0,.25em) to (\i-3.south west);
}

\node[fit=(15-2)(17-2), draw=red, line width=1.5pt, inner sep=.75pt] {};
\node[fit=(15-2)(17-2), draw=black, thin] (hlbox) {};
\node[fit=(15-3)(17-3), draw=red, line width=1.5pt, inner sep=.75pt] {};
\node[fit=(15-3)(17-3), draw=black, thin] (hlbox) {};

\node[below=1em of C1] (sep) {};
\draw[gray, thin] (C1.west |- sep) to (C1.east |- sep);
\node[below=1em of C2] (sep) {};
\draw[gray, thin] (C2.west |- sep) to (C2.east |- sep);

\end{tikzpicture}%
}%
\caption{Illustrations for \cref{lem:per_case}. Gray (hatched) blocks are explicitly (resp., implicitly) bad.}
\end{figure}

%% file: figs/edit.tex
\usetikzlibrary{positioning}

\newcommand{\tablethreebythree}[9]{%
\setlength{\arrayrulewidth}{0.4pt}%
\begin{tabular}{|c|c||c|}
\hline
#1 & #2 & \bf{\textcolor{blue}{#3}} \\ \hline
#4 & #5 & \bf{\textcolor{blue}{#6}} \\ \hline
#7 & #8 & \bf{\textcolor{blue}{#9}} \\ \hline
\end{tabular}%
}
\renewcommand{\arraystretch}{1.2}

\newcommand{\tablethreebyone}[3]{%
\setlength{\arrayrulewidth}{0.4pt}%
\begin{tabular}{|c|}
\hline \bf{\textcolor{blue}{#1}} \\ \hline
\bf{\textcolor{blue}{#2}} \\ \hline
\bf{\textcolor{blue}{#3}} \\ \hline
\end{tabular}%
}


\begingroup
\def\onex{2.0em}
\def\oney{1.8em}
\def\ystep{2.5}

\small

\begin{tikzpicture}[
    treenode/.style={
        circle, draw=black, thick,
        fill=white, inner sep=2pt
    },
    arraynode/.style={
        rectangle, draw=black, thick,
        fill=white,
        minimum width=2.2*\onex,
        minimum height=1.2em
    },
    tablebox/.style={
        fill=white,
        inner sep=1pt,
        font=\scriptsize,
        anchor=center
    },
    x=\onex, y=\oney,
    scale=0.85,
    transform shape
]

\node[arraynode] (p1)  at (0,0) {14};
\node[left=0.3em of p1] {$P=$};
\node[arraynode] (p2)  at (2,0) {15};
\node[arraynode] (p3)  at (4,0) {16};
\node[arraynode] (p4)  at (6,0) {16};
\node[arraynode] (p5)  at (8,0) {12};
\node[arraynode] (p6)  at (10,0) {17};
\node[arraynode] (p7)  at (12,0) {17};
\node[arraynode] (p8)  at (14,0) {18};
\node[arraynode] (p9)  at (16,0) {8};
\node[arraynode] (p10) at (18,0) {19};

\foreach \i in {1,...,10}
  \node[above=0.3em of p\i] {\scriptsize \i};

\foreach \i in {1,...,10}
  \draw[dashed,gray!50] (p\i.south) -- ++(0,-18);

\begin{scope}[yshift=-4.5*\oney]

\node[treenode] (m8)   at (16,0) {9};
\node[tablebox, above=4mm of m8]
{\tablethreebythree{$-\infty$}{$-\infty$}{$-\infty$}{$-\infty$}{1}{$-\infty$}{1}{2}{1}};

\node[treenode] (m12)  at (8,-\ystep) {5};
\node[tablebox, above=4mm of m12]
{\tablethreebythree{$-\infty$}{$-\infty$}{$-\infty$}{$-\infty$}{3}{2}{3}{4}{3}};

\node[treenode] (m19)  at (18,-\ystep) {10};
\node[tablebox, right=4mm of m19]
{\tablethreebyone{9}{$\infty$}{$\infty$}}; 

\node[treenode] (m14)  at (0,-2*\ystep) {1};
\node[tablebox, above=4mm of m14]
{\tablethreebythree{$-\infty$}{$-\infty$}{$-\infty$}{$-\infty$}{5}{4}{5}{6}{5}};

\node[treenode] (m17a) at (10,-2*\ystep) {6};
\node[tablebox, right=4mm of m17a]
{\tablethreebythree{$-\infty$}{7}{7}{7}{8}{7}{8}{$\infty$}{8}};

\node[treenode] (m15)  at (2,-3*\ystep) {2};
\node[tablebox, right=4mm of m15]
{\tablethreebythree{$-\infty$}{$-\infty$}{$-\infty$}{$-\infty$}{6}{5}{6}{$\infty$}{6}};

\node[treenode] (m17b) at (12,-3*\ystep) {7};
\node[tablebox, right=4mm of m17b]
{\tablethreebythree{$-\infty$}{8}{7}{8}{$\infty$}{8}{$\infty$}{$\infty$}{$\infty$}};

\node[treenode] (m16a) at (4,-4*\ystep) {3};
\node[tablebox, right=4mm of m16a]
{\tablethreebythree{$-\infty$}{1}{$-\infty$}{1}{$\infty$}{6}{$\infty$}{$\infty$}{$\infty$}};

\node[treenode] (m18)  at (14,-4*\ystep) {8};
\node[tablebox, right=4mm of m18]
{\tablethreebyone{8}{$\infty$}{$\infty$}}; 

\node[treenode] (m16b) at (6,-5*\ystep) {4};
\node[tablebox, right=4mm of m16b]
{\tablethreebyone{1}{$\infty$}{$\infty$}}; 

\draw (m8)--(m12);
\draw (m8)--(m19);
\draw (m12)--(m14);
\draw (m12)--(m17a);
\draw (m14)--(m15);
\draw (m15)--(m16a);
\draw (m16a)--(m16b);
\draw (m17a)--(m17b);
\draw (m17b)--(m18);

\end{scope}

\begin{scope}[yshift=-20*\oney]

\node[arraynode] (t1)  at (0,0) {4};
\node[left=0.3em of t1] {$T=$};
\node[arraynode] (t2)  at (2,0) {5};
\node[arraynode] (t3)  at (4,0) {6};
\node[arraynode] (t4)  at (6,0) {1};
\node[arraynode] (t5)  at (8,0) {2};
\node[arraynode] (t6)  at (10,0) {7};
\node[arraynode] (t7)  at (12,0) {7};
\node[arraynode] (t8)  at (14,0) {8};
\node[arraynode] (t9)  at (16,0) {3};
\node[arraynode] (t10) at (18,0) {9};

\foreach \i in {1,...,10}
  \node[above=0.3em of t\i] {\scriptsize \i};

\end{scope}

\end{tikzpicture}

\endgroup